\documentclass{article}
\usepackage[text={39pc,50.3pc},centering]{geometry}

\onecolumn  
\usepackage{macros}

\title{Extending the Applicability Condition\\
in the Formal System $\LD$}

\author{Ferruccio Guidi\\
Department of Computer Science and Engineering,\\
University of Bologna, Italy}

\begin{document}

\maketitle

\begin{abstract}

The formal system $\LD$ is a typed $\lambda$-calculus
derived from $\LY$, aiming to support
the foundations of Mathematics that require an underlying theory of expressions
(for example the Minimal Type Theory).

The system is developed in the context of the
Hypertextual Electronic Library of Mathematics
as a machine-checked digital specification,
that is not the formal counterpart of previous informal material.
The first version of the calculus appeared in 2006
and proved unsatisfactory for some reasons.

In this article we present a revised version of the system
and we prove three relevant desired properties:
the confluence of reduction,
the strong normalization of an extended form of reduction,
known as the ``big tree'' theorem,
and the preservation of validity by reduction.
To our knowledge, we are presenting here
the first fully machine-checked proof of the ``big tree'' theorem
for a calculus that includes $\LY$.

\end{abstract}

\section{Introduction}
\SecLabel{introduction}

The formal system $\LD$ is a typed $\lambda$-calculus aiming to support
the foundations of Mathematics that require an underlying theory of expressions
(for example $\MTT$ of \cite{Mai09} and its predecessors).
The system is developed in the context of the HELM project of \cite{APSGS03}
as a machine-checked digital specification,
that is not the formal counterpart of some previous informal material.
The first version of the calculus \cite{lambdadeltaV1a},
formalized in the proof management system (p.m.s.) Coq \cite{CoqOld}
and published by \cite{lambdadeltaJ1a},
proved unsatisfactory for some reasons.
So a revision of the calculus is ongoing since April 2011
and includes a brand new formalization \cite{lambdadeltaV2a}
in the p.m.s. Matita of \cite{ARST11}.
Firstly, the revision aims at this problem:
the calculus of \cite{lambdadeltaJ1a}
comes from $\LY$ \cite{SPAc6}, a language of
the Automath family \cite{SPA94}, and yet
it cannot type every term typed by $\LY$
since it lacks the ``pure'' type inference rule for function application \cite{Bru91}.
If $\NTA{}{\G}{M}{N}$ is a type assignment judgment
and $\NV{}{\G}{M}$ is the corresponding validity judgment, this rule states:
\[
\Rule{\NTA{}{\G}{f}{F}\quad\NV{}{\G}{\MAppl{F}{t}}}
{\mathrm{@-pure}}
{\NTA{}{\G}{\MAppl{f}{t}}{\MAppl{F}{t}}}
\EqLabel{at_pure}\]
This rule is redundant when the terms have three degrees
(objects, classes, and sorts) as in Pure Type Systems (PTS's)
\cite{Brn92} and their derivatives.
On the contrary it becomes effective when more degrees are available,
as in the $\AUTF$ family \cite{SPAb3} or in $\LY$,
since $\NTA{}{\G}{f}{F}$ and $\NTA{}{\G}{F}{\F}$
do not imply that $\F$ is a sort.
In this case $f$ can be a function, $F$ a function space,
and $\F$ a family of function spaces.
If we take $t$ in the domain of $f$,
we might want $\NV{}{\G}{\MAppl{f}{t}}$
even when $f$ and $F$ are given abstractly as variables declared in $\G$.
\Ruleref{at_pure} is designed to realize this situation.

In the mathematical language we express a large variety of concepts,
each with its own requirements.
When we translate this language to typed $\lambda$-calculus,
a widely accepted policy suggests that
expressions denoting concepts with different requirements
should correspond to $\lambda$-terms with different degrees.
Consider typical concepts of interest:
sets, elements, propositions and proofs.
While well-established similarities between elements and proofs
support their representation with terms of the same degree,
significant differences arise as well,
playing in favor of representing them differently.  

Mainly, identifying two proofs of a proposition
(also known as ``proof irrilevance'') is sensible,
while identifying two elements of a set generally is not.
\cite{SPAb3} approaches this problem by advocating a calculus in which
two terms inhabiting the same type of degree $3$ are definitionally
equal. This is to say that terms of degree $4$ are provided for
representing irrelevant proofs. 
Similarly, subtle differences can be found in the requirements for sets and propositions. 
So it seams that a calculus with many degrees for its terms,
may allow flexible interpretations of the mathematical language.

We note that $\LD$ has a disadvantage in this sense
because of its ``isotropy'', 
by which we mean that the features of its terms do not depend on their degree.

Secondly, the revised $\LD$ aims at other improvements
some of which were advocated already by \cite{lambdadeltaJ1a}.
Simpler ``arities'' make the arity assignment judgment
decidable for all values of the sort hierarchy parameter.
The reaxiomatized step of environment-dependent parallel reduction 
allows to remove the substitution operator
and provides for the long-awaited \Ruleref{cpr_appl}.
Tait-style reducibility candidates \cite{Tai75}
in place of Girard-style ones \cite{GTL89}
simplify the strong normalization theorem.
Simpler environments allow to remove some ancillary operators.
\[
\Rule{\CPR{\G}{f_1}{f_2}\quad\CPR{\G}{t_1}{t_2}}
{\mathrm{appl}}
{\CPR{\G}{\MAppl{f_1}{t_1}}{\MAppl{f_2}{t_2}}}
\EqLabel{cpr_appl}\]

The main contributions of this article are
the so-called ``big tree'' theorem \cite{SPAc4} for $\LD$,
which yields the subject reduction theorem for its stratified validity.

The ``big tree'' theorem states that
valid terms are strongly normalizing with respect to a relation 
comprising reduction steps, type steps, subtraction steps, and more.
It generalizes ordinary strong normalization
and gives a very powerful induction principle
for proving properties on valid terms.
We are confident that this tool may prove useful
in systems other than $\LD$ as well.

Stratified validity (\IE validity up to a specified degree)
replaces type assignment as a primitive notion in the revised $\LD$.
This choice is motivated by the subject reduction theorem,
which, in presence of \Ruleref{at_pure}, 
is proved more easily for validity
(the property of having an unspecified type)
than for type assignment (the property of having a specified type) 
since types in $\LD$, as well as in other systems,
are not specified uniquely but up to conversion.
The same situation arises for $\LY$ \cite{SPA2}.

At this stage the revised $\LD$ does not include a type judgment
and the exclusion binder $\chi$ of \cite{lambdadeltaJ1a},
however we conjecture that our notion of validity implies \Ruleref{at_pure}.

The revised $\LD$ is defined in \SecRef{definitions} and its properties
are presented in \SecRef{propositions}. Our conclusions are in \SecRef{conclusion}.
\AppRef{notation} gives a summary of the notation we introduce,
while \AppRef{pointers} gives the pointers to the digital version of our results.

We agree that the symbol \hbox{\hss\qed}
terminates our definitions and our proofs in the text.

\section{Definition of $\LD$}
\SecLabel{definitions}

In this section we define the revised $\LD$ from scratch 
presenting its language (\SubRef{language}),
its reduction rules (\SubRef{cpr_defs}),
and its validity rules (\SubRef{snv_defs}).
These rules depend on some ancillary notions:
relocation (\SubRef{relocation}),
static type assignment (\SubRef{static}),
and degree assignment (\SubRef{da_defs}).
Other notions are introduced to state or prove the main theorems
of this article: closures (\SubRef{closures}),
extended reduction (\SubRef{cpx_defs}),
atomic arity assignment (\SubRef{aaa_defs}),
reducibility candidates (\SubRef{gcp_defs}),
lazy equivalence (\SubRef{lleq_defs}), and
an extension of ``big trees'' termed here ``very big trees'' (\SubRef{fpb_defs}).

We shall use some logical constants:
$\MAll{}$ (universal quantification),
$\MEx{}$ (existential quantification),
$\MImp$ (implication),
$\MAnd$ (conjunction),
and natural numbers with standard operators:
$\le$, $<$, $+$, and $-$. 
We shall need lists for the normalization theorem.
Metavariables for lists will be overlined, like $\vect{c}$.
The empty list will be $\Empty$, and the infix semicolon will denote
concatenation, like $\Cons{c}{\vect{c}}$.

Contrary to \cite{lambdadeltaJ1a},
in this presentation we want to
follow the digital specification of the calculus strictly,
especially in the treatment of variables,
and we make some notational changes with respect to that article.
The reader will find a summary of the revised notation in \AppRef{notation}.

\subsection{Language}
\SubLabel{language}


The grammar of $\LD$ features two syntactic categories:
terms and environments, and uses natural numbers.
Terms are presented in the ``item notation'' of \cite{KN96b},
and include sorts, variable occurrences,
abbreviations, typed abstractions,
applications, and type annotations.
Contrary to \cite{lambdadeltaJ1a},
environments contain just (nonrecursive) definitions, and typed declarations.

\begin{definition}[terms and environments]
\ObjLabel{tl}
Terms and environments are defined in \FigRef{tl}{}.
$\Sort{k}$ is the sort of index $k$,
$\LRef{i}$ is the reference to the variable introduced at depth $i$
\cite{SPAc2} (so $i$ is a ``de Bruijn index''),
$\Abbr{}{V}{T}$ is the abbreviation ``$\MAbbr{\LRef{0}}{V}{T}$'',
$\Abst{}{W}{T}$ is the function ``$\MAbst{\LRef{0}}{W}{T}$'',
$\Appl{V}{T}$ is the application ``$\MAppl{T}{V}$'',
and $\Cast{W}{T}$ is the type annotation ``$\MCast{T}{W}$''.
$\Null$ is the empty environment,
$\LDef{L}{}{V}$ is $L$ with the definition ``$\MAbbr{\LRef{0}}{V}{}$'',
and $\LDec{L}{}{W}$ is $L$ with the declaration ``$\MAbst{\LRef{0}}{W}{}$''.
\qed
\end{definition}

\begin{Figure}{Terms and environments}{tl}
\begin{tabular}{lrl}
natural number & $i,j,k$ & starting at $0$\\
term & $T, U, V, W$ & 
$\GLet 
 \Sort{k} \GOr \LRef{i} \GOr
 \Abbr{}{V}{T} \GOr \Abst{}{W}{T} \GOr
 \Appl{V}{T} \GOr \Cast{W}{T}
$\\
environment & $K,L$ &
$\GLet
 \Null \GOr \LDef{L}{}{V} \GOr \LDec{L}{}{W}
$\\
\end{tabular}
\end{Figure}

Convention:
the symbol $\EITHER{\delta}{\lambda}$
means: ``either $\delta$, or $\lambda$''.
If the symbol occurs many times in a statement, it means:
``either $\delta$ in every occurrence, or $\lambda$ in every occurrence''.
The same convention holds for similar symbols we will use,
like $\EITHER{\star}{\#}$ and $\EITHER{\copyright}{@}$.


The application can be extended to take a list $\vect{V}$ of arguments.

\begin{definition}[multiple application]
\ObjLabel{applv}
$\ApplV{\vect{V}}{}{T}$ defined in \FigRef{applv}{},
denotes the application of $T$ the arguments in the list $\vect{V}$
starting from the rightmost term in $\vect{V}$. 
\qed
\end{definition}

\begin{Figure}{Multiple application}{applv}
$\ApplV{\Empty}{}{T} = T$ \quad
$\ApplV{(\Cons{V}{\vect{V}})}{}{T} = \Appl{V}{(\ApplV{\vect{V}}{}{T})}$
\end{Figure}


Environments are lists
so some standard operators can be defined on them.

\begin{definition}[length]
\ObjLabel{length}
\FigRef{length}{} defines the length $\Length{L}$ of an environment $L$.
\qed
\end{definition}

\begin{Figure}{Length of an environment}{length}
$\Length{\Null} = 0$ \quad
$\Length{\Pair{L}{}{W}} = \Length{L}$
\end{Figure}


\begin{definition}[concatenation]
\ObjLabel{append}
\FigRef{append}{} defines 
the concatenation $\Append{K}{L}$ of $L$ before $K$.
In particular we write $\Bind{}{W}{L}$ for $\Append{(\Pair{\Null}{}{W})}{L}$.
\qed
\end{definition}

\begin{Figure}{Concatenation of two environments}{append}
$\Append{K}{\Null} = K$ \quad
$\Append{K}{(\Pair{L}{}{W})} = \Pair{(\Append{K}{L})}{}{W}$
\end{Figure}


Normalization requires two predicates:
see \Defref{gcr}{} and \Thref{cpx_props}{cpxs_fwd_beta}.

\begin{definition}[neutrality]
\ObjLabel{simple}
$\Simple{T}$ states that the term $T$ is simple (or
neutral) as defined in \FigRef{simple}{}.
Specifically, $T$ is neither an abbreviation, nor an abstraction.
\qed
\end{definition}

\begin{Figure}{Simple (or neutral) terms}{simple}
$\Simple{\Atom{i}}$ \quad $\Simple{\Flat{V}{T}}$
\end{Figure}


\begin{definition}[top structure]
\ObjLabel{tsts}
$\TSTS{T_1}{T_2}$ states that
the terms $T_1$ and $T_2$ have the same top structure
as defined in \FigRef{tsts}{}.
Specifically, $T_1$ and $T_2$ are the same atomic term
or start with the same constructor.
\qed
\end{definition}

\begin{Figure}{Terms with the same top structure}{tsts}
$\TSTS{\Atom{i}}{\Atom{i}}$ \quad
$\TSTS{\Item{}{V_1}{T_1}}{\Item{}{V_2}{T_2}}$
\end{Figure}

\subsection{Relocation}
\SubLabel{relocation}


Managing variables referred by depth requires
a well-known function $\Lift{l}{m}{T}{}$
connected to the function $\MAppl{\tau_m}{T}$ of \cite{SPAc2}.
In particular, when the term $T$ enters the scope of $m$ binders,
$\Lift{0}{m}{T}{}$ relocates the indexes of its free variables.  
The composition of such functions is of interest as well.

\begin{definition}[relocation]
\ObjLabel{lift}
The relation $\Lift{l}{m}{T_1}{T_2}$ defined in \FigRef{lift}{},
states that $T_2$ is the relocation of $T_1$
at level $l$ with depth (or ``height'') $m$.

We term the pair $\CLOSURE{l}{m}$ a ``relocation pair''.
\qed
\end{definition}

\begin{Figure}{Relocation}{lift}
\begin{tabular}{lrl}
natural number & $l,m$ & starting at $0$\\
relocation pair & $c$ & $\GDef \CLOSURE{l}{m}$\\
\end{tabular}\\[0.5pc]
\Rule{}{\mathrm{sort}}{\Lift{l}{m}{\Sort{k}}{\Sort{k}}}\quad
\Rule{i < l}{\mathrm{lref\_lt}}{\Lift{l}{m}{\LRef{i}}{\LRef{i}}}\quad
\Rule{l \le i}{\mathrm{lref\_ge}}{\Lift{l}{m}{\LRef{i}}{\LRef{(i+m)}}}\NL
\Rule{\Lift{l}{m}{W_1}{W_2}\quad\Lift{l+1}{m}{T_1}{T_2}}
     {\mathrm{bind}}
     {\Lift{l}{m}{\Bind{}{W_1}{T_1}}{\Bind{}{W_2}{T_2}}}\quad
\Rule{\Lift{l}{m}{V_1}{V_2}\quad\Lift{l}{m}{T_1}{T_2}}
     {\mathrm{flat}}
     {\Lift{l}{m}{\Flat{V_1}{T_1}}{\Flat{V_2}{T_2}}}
\end{Figure}


\begin{definition}[vector relocation]
\ObjLabel{liftv}
The relation $\LiftV{l}{m}{\vect{T}_1}{\vect{T}_2}$ defined in \FigRef{liftv}{},
applies $\CLOSURE{l}{m}$ to the components of the list $\vect{T}_1$
preserving their order in the list $\vect{T}_2$.
\qed
\end{definition}

\begin{Figure}{Vector relocation}{liftv}
\Rule{}{\mathrm{empty}}{\LiftV{l}{m}{\Empty}{\Empty}} \quad
\Rule{\Lift{l}{m}{T_1}{T_2} \quad \LiftV{l}{m}{\vect{T}_1}{\vect{T}_2}}
     {\mathrm{cons}}
     {\LiftV{l}{m}{(\Cons{T_1}{\vect{T}_1})}{\Cons{T_2}{\vect{T}_2}}}
\end{Figure}


\begin{definition}[multiple relocation]
\ObjLabel{lifts}
The relation $\Lifts{\vect{c}}{T_1}{T_2}$ defined in \FigRef{lifts}{},
applies the list $\vect{c}$ of relocation pairs to $T_1$
starting from the leftmost pair in $\vect{c}$.
\qed
\end{definition}

\begin{Figure}{Multiple relocation}{lifts}
\Rule{}{\mathrm{empty}}{\Lifts{\Empty}{T}{T}} \quad
\Rule{\Lifts{c}{T_1}{T} \quad \Lifts{\vect{c}}{T}{T_2}}
     {\mathrm{cons}}
     {\Lifts{\Cons{c}{\vect{c}}}{T_1}{T_2}}
\end{Figure}


If $\Lift{l}{m}{T_1}{T_2}$, notably, $T_2$ does not refer to
the variables introduced at depth $i$ with $l \le i < l + m$. 
So a relation $\Drop{l}{m}{L_1}{L_2}$ is provided
for removing the $i$-th entries of $L_1$ such that $l \le i < l + m$,
while relocating the $i$-th entries such that $i < l$.
The relation is defined only when this relocation is possible, 
that is when an $i$-th entry with $i < l$
does not refer to an $i$-th entry with $l \le i < l + m$.
The $0-th$ entry of $L_1$ is the head of $L_1$.
We term this relation ``drop''
as opposed to relocation, which is sometimes termed ``lift''.

Notice that if $\Drop{0}{i}{L_1}{L_2}$,
then the head of $L_2$ contains the $i$-th entry of $L_1$.

\begin{definition}[drop]
\ObjLabel{drop}
The relation $\Drop{l}{m}{L_1}{L_2}$ defined in \FigRef{drop}{},
states that $L_2$ is $L_1$ without the $i$-th entries such that $l \le i < l + m$,
and with the $i$-th entries such that
$i < l$ relocated accordingly.
\qed
\end{definition}

\begin{Figure}{Drop}{drop}
\Rule{}{\mathrm{atom}}{\Drop{l}{0}{\Null}{\Null}} \quad
\Rule{\Drop{0}{0}{L_1}{L_2}}{\mathrm{pair}}
     {\Drop{0}{0}{\Pair{L_1}{}{W}}{\Pair{L_2}{}{W}}} \NL
\Rule{\Drop{0}{m}{L_1}{L_2}}{\mathrm{drop}}
     {\Drop{0}{m+1}{\Pair{L_1}{}{W}}{L_2}} \quad
\Rule{\Drop{l}{m}{L_1}{L_2}\quad\Lift{l}{m}{W_2}{W_1}}{\mathrm{skip}}
     {\Drop{l+1}{m}{\Pair{L_1}{}{W_1}}{\Pair{L_2}{}{W_2}}}
\end{Figure}

\FigRef{drop}{atom} generalizes ``drop'' of \cite{lambdadeltaV1a}
allowing $\Drop{l}{0}{L}{L}$ when $\Length{L} \le l$.


\begin{definition}[multiple drop]
\ObjLabel{drops}
The relation $\Drops{\vect{c}}{L_1}{L_2}$ defined in \FigRef{drops}{},
applies the list $\vect{c}$ of relocation pairs to $L_1$
starting from the leftmost pair in $\vect{c}$.
\qed
\end{definition}

\begin{Figure}{Multiple Drop}{drops}
\Rule{}{\mathrm{empty}}{\Drops{\Empty}{L}{L}} \quad
\Rule{\Drops{c}{L_1}{L} \quad \Drops{\vect{c}}{L}{L_2}}
     {\mathrm{cons}}
     {\Drops{\Cons{c}{\vect{c}}}{L_1}{L_2}}
\end{Figure}


The next equivalence relation
appears in \Thref{lleq_props}{lleq_lreq_trans}.

\begin{definition}[ranged equivalence]
\ObjLabel{lreq}
The relation $\LREq{l}{m}{L_1}{L_2}$ defined in \FigRef{lreq}{},
states that $L_1$ and $L_2$ have the same length
and the same $i$-th entries for $l \le i < l + m$.
\qed
\end{definition}

\begin{Figure}{Ranged equivalence}{lreq}
\Rule{}{\mathrm{atom}}{\LREq{l}{m}{\Null}{\Null}} \quad
\Rule{\LREq{0}{m}{L_1}{L_2}}{\mathrm{pair}}
     {\LREq{0}{m+1}{\Pair{L_1}{}{W}}{\Pair{L_2}{}{W}}} \NL
\Rule{\LREq{0}{0}{L_1}{L_2}}{\mathrm{zero}}
     {\LREq{0}{0}{\Pair{L_1}{1}{W_1}}{\Pair{L_2}{2}{W_2}}} \quad
\Rule{\LREq{l}{m}{L_1}{L_2}}{\mathrm{succ}}
     {\LREq{l+1}{m}{\Pair{L_1}{1}{W_1}}{\Pair{L_2}{2}{W_2}}}
\end{Figure}

\subsection{Reduction}
\SubLabel{cpr_defs}

$\LD$ features a transition system with five schemes of reducible expressions (redexes).
Care is taken to design a deterministic and confluent system
with disjoint redex schemes,
in which the call-by-value $\beta$-reduction is broken into its basic components.

\begin{definition}[transitions]
\ObjLabel{transition}
\FigRef{transition}{} defines the redexes
and their transitions $\beta$, $\delta$, $\epsilon$,
$\zeta$, and $\theta$, which depend on an environment $L$.
The $\beta$-reduction is delayed (call-by-name style),
the $\delta$-expansion expands a definition in $L$,
the $\epsilon$-contraction removes a type annotation,
the $\zeta$-contraction removes an unreferenced abbreviation,
and the $\theta$-reduction \cite{CH00}
swaps an application and an abbreviation.
\qed
\end{definition}

\begin{Figure}{Transitions}{transition}
\Rule{}{\beta}
     {\CQR{L}{\Appl{V}{\Abst{}{W}{T}}}{\Abbr{}{(\Cast{W}{V})}{T}}} \quad
\Rule{\Drop{0}{i}{L}{\LDef{K}{}{V_1}}\quad\Lift{0}{i+1}{V_1}{V_2}}{\delta}
     {\CQR{L}{\LRef{i}}{V_2}} \NL
\Rule{}{\epsilon}{\CQR{L}{\Cast{U}{T}}{T}} \quad
\Rule{\Lift{0}{1}{T_2}{T_1}}{\zeta}{\CQR{L}{\Abbr{}{V}{T_1}}{T_2}} \quad
\Rule{\Lift{0}{1}{V_1}{V_2}}{\theta}
     {\CQR{L}{\Appl{V_1}{\Abbr{}{W}{T}}}{\Abbr{}{W}{\Appl{V_2}{T}}}}
\end{Figure}

Notice that the $\beta$-redex contains a type annotation $W$ that,
contrary to \cite{lambdadeltaJ1a}, remains in the $\beta$-reductum.
This choice is connected with the revised form of the normalization theorem.
Also notice that $\delta$-expansion, contrary to \cite{lambdadeltaJ1a},
does not mention substitution.
In the light of next \DefRef{cpr}{},
delayed parallel substitution is seen as a special case of reduction.


Following \cite{lambdadeltaJ1a}, we present parallel reduction
to ease the proof of the confluence theorem,
but here we take environment-dependent reduction as primitive.

\begin{definition}[parallel reduction for terms]
\ObjLabel{cpr}
The relation $\CPR{L}{T_1}{T_2}$ defined in \FigRef{cpr}{}, indicates
one step of parallel reduction from $T_1$ to $T_2$ in $L$.
\qed
\end{definition}

\begin{Figure}{Parallel reduction for terms (single step)}{cpr}
\Rule{\CPR{L}{W_1}{W_2}\quad\CPR{\Pair{L}{}{W_1}}{T_1}{T_2}}{\mathrm{bind}}
     {\CPR{L}{\Bind{}{W_1}{T_1}}{\Bind{}{W_2}{T_2}}} \quad
\Rule{\CPR{L}{V_1}{V_2}\quad\CPR{L}{T_1}{T_2}}{\mathrm{flat}}
     {\CPR{L}{\Flat{V_1}{T_1}}{\Flat{V_2}{T_2}}} \NL
\Rule{}{\mathrm{atom}}{\CPR{L}{\Atom{i}}{\Atom{i}}} \quad
\Rule{\Drop{0}{i}{L}{\LDef{K}{}{V_1}}\quad\CPR{K}{V_1}{V_2}\quad\Lift{0}{i+1}{V_2}{W_2}}{\delta}
     {\CPR{L}{\LRef{i}}{W_2}} \NL
\Rule{\CPR{L}{V_1}{V_2}\quad\CPR{L}{W_1}{W_2}\quad\CPR{\LDec{L}{}{W_1}}{T_1}{T_2}}{\beta}
     {\CPR{L}{\Appl{V_1}{\Abst{}{W_1}{T_1}}}{\Abbr{}{(\Cast{W_2}{V_2})}{T_2}}} \quad
\Rule{\CPR{\LDef{L}{}{V}}{U_1}{U_2}\quad\Lift{0}{1}{T_2}{U_2}}{\zeta}
     {\CPR{L}{\Abbr{}{V}{U_1}}{T_2}} \NL
\Rule{\CPR{L}{T_1}{T_2}}{\epsilon}{\CPR{L}{\Cast{U}{T_1}}{T_2}} \quad
\Rule{\CPR{L}{V_1}{V_2}\quad\Lift{0}{1}{V_2}{W_2}\quad\CPR{L}{U_1}{U_2}\quad\CPR{\LDef{L}{}{U_1}}{T_1}{T_2}}{\theta}
     {\CPR{L}{\Appl{V_1}{\Abbr{}{U_1}{T_1}}}{\Abbr{}{U_2}{\Appl{W_2}{T_2}}}}
\end{Figure}

We compute a call-by-value $\beta$-reduction in two steps,
as we illustrate by computing the term $\MAppl{\Delta}{\Delta}$.
In particular we set $\Delta_T = \Abst{}{T}{\Appl{\LRef{0}}{\LRef{0}}}$
and we agree that $\Lift{0}{1}{T}{U}$.

\begin{center}\begin{tabular}{rl}
$\beta$&
$\CPR{L}{\Appl{\Delta_T}{\Delta_T}}{\Abbr{}{(\Cast{T}{\Delta_T})}{\Appl{\LRef{0}}{\LRef{0}}}}$\\
$\epsilon$, $\delta$, $\zeta$&
$\CPR{L}{\Abbr{}{(\Cast{T}{\Delta_T})}{\Appl{\LRef{0}}{\LRef{0}}}}{\Appl{\Delta_T}{\Delta_T}}$\\
&by $\CPR{\LDef{L}{}{(\Cast{T}{\Delta_T}})}{\LRef{0}}{\Delta_U}$
and $\CPR{L}{\Cast{T}{\Delta_T}}{\Delta_T}$\\
\end{tabular}\end{center}

The advantage of environment-dependent parallel reduction
over the approach of \cite{lambdadeltaJ1a}
lies in the increased parallelism of $\delta$-expansions,
which we need for the ``big tree'' theorem.
Suppose that $\Subst{m}{V}{T}{}$ replaces with $V$ some references
in $T$ to the variable introduced at depth $m$,
and compare \FigRef{cpr}{bind} and \FigRef{cpr}{\delta}
with \Ruleref{deltafree} (\IE their environment-free counterpart).
When we replace many variable instances in one step with this rule,
each instance receives the same reduct $V_2$ of $V_1$.
Whereas, by \FigRef{cpr}{\delta} each instance may receive a different
reduct of $V_1$.
\[
\Rule{\CPR{}{V_1}{V_2}\quad\CPR{}{T_1}{T}\quad\Subst{0}{V_2}{T}{T_2}}
{\mathrm{\delta-free}}
{\CPR{}{\Abbr{}{V_1}{T_1}}{\Abbr{}{V_2}{T_2}}}
\EqLabel{deltafree}\]

Notice that the subsystem of rules:
\FigRef{cpr}{bind}, \FigRef{cpr}{flat}, \FigRef{cpr}{atom}, and \FigRef{cpr}{\delta}
axiomatizes environment-dependent parallel substitution.


We derive several notions from parallel reduction:
an extension for environments needed in the confluence theorem,
and some transitive closures.
In this setting we agree that
a ``computation'' is a reduction sequence consisting of zero or more steps.

\begin{definition}[parallel reduction for environments]
\ObjLabel{lpr}
The relation $\LPR{L_1}{L_2}$ defined in \FigRef{lpr}{} indicates
one step of parallel reduction from $L_1$ to $L_2$.
\qed
\end{definition}

\begin{Figure}{Parallel reduction for environments (single step)}{lpr}
\Rule{}{\mathrm{atom}}{\LPR{\Null}{\Null}} \quad
\Rule{\LPR{L_1}{L_2}\quad\CPR{L_1}{W_1}{W_2}}{\mathrm{pair}}
     {\LPR{\Pair{L_1}{}{W_1}}{\Pair{L_2}{}{W_2}}}
\end{Figure}


\begin{definition}[parallel computation and conversion]
\ObjLabel{cprs}
The relation $\CPRS{L}{T_1}{T_2}$ (computation) is the transitive closure of $\CPR{L}{T_1}{T_2}$,
while $\CPCS{L}{T_1}{T_2}$ (conversion) is the symmetric and transitive closure of $\CPR{L}{T_1}{T_2}$.
Moreover $\LPRS{L_1}{L_2}$ (computation) is the transitive closure of $\LPR{L_1}{L_2}$.
\FigRef{cprs}{} defines $\CPRS{L}{T_1}{T_2}$ for reference.
The other notions are defined in the same manner.
\qed
\end{definition}

\begin{Figure}{Parallel computation for terms (multi-step)}{cprs}
\Rule{\CPR{L}{T_1}{T_2}}{\mathrm{inj}}{\CPRS{L}{T_1}{T_2}} \quad
\Rule{\CPRS{L}{T_1}{T}\quad\CPR{L}{T}{T_2}}{\mathrm{step}}
     {\CPRS{L}{T_1}{T_2}}
\end{Figure}

The transitive closures we just defined are reflexive,
because so is $\CPR{L}{T_1}{T_2}$.
Therefore the symbol $^*$ in their notation is justified
as a Kleene star meaning ``zero or more''.

A characteristic feature of $\LD$ is the use of reflexive
relations for environments termed here ``refinements'',
invoked when proving that reduction preserves some property.
Specifically, they are invoked in the case of \FigRef{cpr}{\beta}
given that a backward application of \FigRef{cpr}{bind} moves
part of the $\beta$-redex and part of the $\beta$-reductum in the environment.
The basic refinement is given next and
occurs in the proof of the confluence theorem.
The other refinements imply this one.
See \DefRef{lsubd}{}, \DefRef{lsubsv}{},
\DefRef{lsuba}{}, \DefRef{lsubc}{}.

\begin{definition}[refinement for preservation of reduction]
\ObjLabel{lsubr}
\FigRef{lsubr}{} defines the relation $\LSubR{L_1}{L_2}$ stating that
$L_1$ refines $L_2$ for preservation of reduction.
\qed
\end{definition}

\begin{Figure}{Refinement for preservation of reduction}{lsubr}
\Rule{}{\mathrm{atom}}{\LSubR{L}{\Null}} \quad
\Rule{\LSubR{L_1}{L_2}}{\mathrm{pair}}
     {\LSubR{\Pair{L_1}{}{W}}{\Pair{L_2}{}{W}}} \quad
\Rule{\LSubR{L_1}{L_2}}{\mathrm{beta}}
     {\LSubR{\LDef{L_1}{}{(\Cast{W}{V})}}{\LDec{L_2}{}{W}}}
\end{Figure}

The main results on reduction, conversion, and refinement
are in \SubRef{cpr_props}.

\subsection{Iterated Static Type Assignment}
\SubLabel{static}


The ``static'' type assignment defined in this section is our counterpart
of the so-called ``de Bruijn'' type assignment
of the Automath tradition \cite{SPA2}.
As such, it plays a central role in our definition of validity. 
Its name recalls that we can compute it
without the help of $\beta\zeta\theta$-reductions.

Intuitively, the term $T$ has a static type $U$ in the environment $L$
iff the head variable occurrence of $T$ is hereditarily closed in $L$.
In that case, $U$ is just a candidate type for $T$.
However, when $T$ is valid, its static type serves
as the ``canonical'' type \cite{KN96a},
or as the ``inferred'' type \cite{Cos96}.

The ``static type iterated $n$ times'' is related to
the notion of validity implied by \Ruleref{at_pure}
and it will be convenient to define it 
as a primitive notion (denoted by $\LSTAS{}{L}{T}{U}{n}$),
that will not be the reflexive and transitive closure
of the ``static type iterated one time''.
In fact we are not interested in full reflexivity
(\IE $\LSTAS{}{L}{T}{T}{0}$ for each $T$).
On the contrary, we wish to ensure that $\LSTAS{}{L}{T}{U}{n}$
holds iff the head variable occurrence of $T$ is hereditarily closed in $L$
regardless of $n$, hence even for $n=0$.
As a matter of fact, differentiating the case $n=0$ for the sake of reflexivity,
yields a less elegant definition of $\LSTAS{}{L}{T}{U}{n}$.

According to our type policy,
the sort of index $k$ is typed by the sort of index $\Next{h}{k}$
where $\Next{h}{}$ is function chosen at will
as long as an extension condition is satisfied.

\begin{definition}[iterated static type assignment]
\ObjLabel{lstas}
A ``sort hierarchy parameter'' is any function $\Next{h}{}$
satisfying the strict extension condition: $k < \Next{h}{k}$.
Moreover $\MIter{n}{h}$ will denote $h$ composed $n$ times.
For a natural number $n$,
the relation $\LSTAS{h}{L}{T}{U}{n}$ defined in \FigRef{lstas}{},
indicates that $U$ is the $n$-iterated ``static'' type of $T$
in $L$ according to $\Next{h}{}$.
\qed
\end{definition}

\begin{Figure}{Iterated static type assignment}{lstas}
\begin{tabular}{lrl}
natural number & $n$ & starting at $0$
\end{tabular}\\[0.5pc]
\Rule{}{\mathrm{sort}}{\LSTAS{h}{L}{\Sort{k}}{\Sort{(\Iter{n}{h}{k})}}{n}}\quad
\Rule{\Drop{0}{i}{L}{\LDec{K}{}{W}}\quad\LSTAS{h}{K}{W}{V}{0}}
     {\mathrm{zero}}{\LSTAS{h}{L}{\LRef{i}}{\LRef{i}}{0}}\NL
\Rule{\LSTAS{h}{\Pair{L}{}{W}}{T}{U}{n}}{\mathrm{bind}}
     {\LSTAS{h}{L}{\Bind{}{W}{T}}{\Bind{}{W}{U}}{n}}\NL
\Rule{\Drop{0}{i}{L}{\LDec{K}{}{W_1}}\quad\LSTAS{h}{K}{W_1}{V_1}{n}\quad\Lift{0}{i+1}{V_1}{V_2}}
     {\mathrm{succ}}{\LSTAS{h}{L}{\LRef{i}}{V_2}{n+1}}\quad
\Rule{\LSTAS{h}{L}{T}{U}{n}}{\mathrm{appl}}{\LSTAS{h}{L}{\Appl{V}{T}}{\Appl{V}{U}}{n}}\NL
\Rule{\Drop{0}{i}{L}{\LDef{K}{}{V}}\quad\LSTAS{h}{K}{V}{W_1}{n}\quad\Lift{0}{i+1}{W_1}{W_2}}
     {\mathrm{ldef}}{\LSTAS{h}{L}{\LRef{i}}{W_2}{n}}\quad
\Rule{\LSTAS{h}{L}{T}{U}{n}}{\mathrm{cast}}{\LSTAS{h}{L}{\Cast{W}{T}}{U}{n}}
\end{Figure}

This definition allows to say that $U$ is the static type of $T$ in $L$
when $\LSTAS{h}{L}{T}{U}{1}$, which differs in \FigRef{lstas}{cast} from the notion
$\STA{h}{L}{T}{U}$ defined by \cite{lambdadeltaJ1a} with the name $\mathrm{st}$.
For example we have
$\STA{h}{\LDec{L}{}{(\Cast{\Sort{k_1}}{\Sort{k_2}})}}{\LRef{0}}{(\Cast{\Sort{k_1}}{\Sort{k_2}})}$
but
$\LSTAS{h}{\LDec{L}{}{(\Cast{\Sort{k_1}}{\Sort{k_2}})}}{\LRef{0}}{\Sort{k_2}}{1}$.
Although $\LSTAS{h}{L}{T}{T}{0}$ does not hold in general,
we can prove that $\LSTAS{h}{L}{T_1}{T_2}{0}$ implies $\CPR{L}{T_1}{T_2}$
by $\delta$-expansion and $\epsilon$-contraction.
We remark that the rules of \FigRef{lstas}{} are syntax-oriented,
so the $n$-iterated static type of $T$ in $L$,
if it exists, is unique for any given $h$ and $n$.
See \Thref{lstas_props}{lstas_mono}.

\subsection{Degree Assignment}
\SubLabel{da_defs}

The ``degree'' of a term $T$ is a number $d$
indicating the position of $T$ in a type hierarchy.
A well-established definition assigns degree $1$ to the bottom sort
(for instance $\tau$ in $\LY$ or $\Sort{}$ in the $\LCube$ \cite{Brn92})
and degree $d+1$ to $T$ such that $\NTA{}{\G}{T}{U}$ when $U$ has degree $d$. 
In $\LD$, as in $\ECC$ \cite{Luo90}, there is no top sort and the degree is an integer.
So this definition prevents from reasoning by induction on the degree.


According to our policy,
the degree of a sort is a natural number given by a function $g$
that can be chosen at will as long as a compatibility condition is satisfied.

Once sorts are assigned a degree, the assignment extends to terms
accordingly.

\begin{definition}[degree assignment]
\ObjLabel{da}
Given a sort hierarchy parameter $\Next{h}{}$,
a ``sort degree parameter'' is any function $\Deg{h}{g}{}{}$
satisfying the compatibility condition:
if $\Deg{h}{g}{k}{d}$ then $\Deg{h}{g}{\Next{h}{k}}{d-1}$.
The relation $\DA{h}{g}{L}{T}{d}$ defined in \FigRef{da}{},
indicates that $T$ has degree $d$ 
in $L$ according to $\Next{h}{}$ and $\Deg{h}{g}{}{}$. 
\qed
\end{definition}

\begin{Figure}{Degree assignment}{da}
\begin{tabular}{lrl}
natural number & $d$ & starting at $0$\\
\end{tabular}\\[0.5pc]
\Rule{\Deg{h}{g}{k}{d}}{\mathrm{sort}}{\DA{h}{g}{L}{\Sort{k}}{d}}\quad
\Rule{\Drop{0}{i}{L}{\LDef{K}{}{V}}\quad\DA{h}{g}{K}{V}{d}}{\mathrm{ldef}}
     {\DA{h}{g}{L}{\LRef{i}}{d}}\quad
\Rule{\Drop{0}{i}{L}{\LDec{K}{}{W}}\quad\DA{h}{g}{K}{W}{d}}{\mathrm{ldec}}
     {\DA{h}{g}{L}{\LRef{i}}{d+1}}\NL
\Rule{\DA{h}{g}{\Pair{L}{}{W}}{T}{d}}{\mathrm{bind}}
     {\DA{h}{g}{L}{\Bind{}{W}{T}}{d}}\quad
\Rule{\DA{h}{g}{L}{T}{d}}{\mathrm{flat}}
     {\DA{h}{g}{L}{\Flat{V}{T}}{d}}
\end{Figure}

As we see, the term $T$ has a degree in $L$
iff the head variable occurrence of $T$ is hereditarily closed in $L$.
So having a degree, is equivalent to having a static type.


The refinement given next occurs in the proof of the preservation theorem
and is needed to prove that the reduction of valid terms preserves their degree.

\begin{definition}[refinement for preservation of degree]
\ObjLabel{lsubd}
\FigRef{lsubd}{} defines the relation $\LSubD{h}{g}{L_1}{L_2}$ stating that
$L_1$ refines $L_2$ for preservation of degree.
\qed
\end{definition}

\begin{Figure}{Refinement for preservation of degree}{lsubd}
\Rule{}{\mathrm{atom}}{\LSubD{h}{g}{\Null}{\Null}} \quad
\Rule{\LSubD{h}{g}{L_1}{L_2}}{\mathrm{pair}}
     {\LSubD{h}{g}{\Pair{L_1}{}{W}}{\Pair{L_2}{}{W}}} \NL
\Rule{\LSubD{h}{g}{L_1}{L_2}\quad\DA{h}{g}{L_1}{V}{d+1}\quad\DA{h}{g}{L_2}{W}{d}}{\mathrm{beta}}
     {\LSubD{h}{g}{\LDef{L_1}{}{(\Cast{W}{V})}}{\LDec{L_2}{}{W}}}
\end{Figure}

The main results on degree assignment and on its refinement
are in \SubRef{da_props}.

\subsection{Stratified Validity}
\SubLabel{snv_defs}

Our validity rules for a term, say $X$, in an environment $L$,
are designed to ensure that:

\begin{enumerate}

\item
a variable occurrence is closed in $L$;
the expected type of a declared variable occurrence
is valid in its environment;
the expansion of a defined variable occurrence
is valid in its environment;

\item
every subterm of $X$ is valid in its environment;

\item\ObjLabel{castability}
for a type annotation $\Cast{W}{V}$,
the inferred type of $V$ converts to $W$ in $L$; 

\item\ObjLabel{applicability}
for an application $\Appl{V}{T}$,
the inferred type of $T$ iterated enough times
converts to the form $\Abst{}{W}{U}$, and 
the inferred type of $V$ converts to $W$ in $L$.

\end{enumerate}

\Thref{}{applicability} is our extension of
the ``applicability condition'', which in a PTS is:

\begin{itemize}

\item
for an application $\Appl{V}{T}$,
the inferred type of $T$ iterated one time
converts to the form $\Prod{}{W}{U}$, and 
the inferred type of $V$ converts to $W$ in $L$.

\end{itemize}

In \cite{lambdadeltaJ1a}
we took by mistake the latter condition replacing $\Pi$ with $\lambda$,
rather than \Thref{}{applicability}.
The idea of \Thref{}{castability} and \Thref{}{applicability}
is that a valid term is typable and its types are the valid terms
that convert to its inferred type.
Notice that this property holds for the calculus of \cite{lambdadeltaJ1a}.
As for $\LY$ \cite{SPA2},
the preservation theorem for $\LD$
(stating that validity is preserved by reduction)
requires an induction on the degree
motivated by its extended applicability condition.

So we define a ``stratified'' validity
depending on a degree assignment
in that we require a positive degree for $V$ in 
\Thref{}{castability} and \Thref{}{applicability},
and in that the inferred type of $T$ is not iterated
more times than the degree of $T$ in \Thref{}{applicability}.
Intuitively, this is validity up to a degree.
%
%
The next ancillary relations are needed
in the formal statement of \Thref{}{castability}
and \Thref{}{applicability}.

\begin{definition}[decomposed computation and conversion]
\ObjLabel{scpds}
the relation $\SCPDS{h}{g}{L}{T_1}{T_2}{n}$
defined in \FigRef{scpds}{}, 
concatenates a degree-guarded iterated static type assignment and a computation,
and the corresponding conversion $\SCPES{h}{g}{L}{T_1}{T_2}{n_1}{n_2}$.
\qed
\end{definition}

\begin{Figure}{Stratified decomposed computation and conversion}{scpds}
\Rule{n \le d\quad\DA{h}{g}{L}{T_1}{d}\quad\LSTAS{h}{L}{T_1}{T}{n}\quad\CPRS{L}{T}{T_2}}
     {\mathrm{scpds}}{\SCPDS{h}{g}{L}{T_1}{T_2}{n}}\NL
\Rule{\SCPDS{h}{g}{L}{T_1}{T}{n_1}\quad\SCPDS{h}{g}{L}{T_2}{T}{n_2}}
     {\mathrm{scpes}}{\SCPES{h}{g}{L}{T_1}{T_2}{n_1}{n_2}}
\end{Figure}


\begin{definition}[stratified validity]
\ObjLabel{snv}
The relation $\SNV{h}{g}{L}{T}$ defined in \FigRef{snv}{}
states that the term $T$ is valid in $L$
with respect to the parameters $h$ and $\Deg{h}{g}{}{}$.
\qed
\end{definition}

\begin{Figure}{Stratified validity}{snv}
\Rule{}{\mathrm{sort}}{\SNV{h}{g}{L}{\Sort{k}}}\quad
\Rule{\Drop{0}{i}{L}{\Pair{K}{}{W}}\quad\SNV{h}{g}{K}{W}}
     {\mathrm{lref}}{\SNV{h}{g}{L}{\LRef{i}}}\quad
\Rule{\SNV{h}{g}{L}{W}\quad\SNV{h}{g}{\Pair{L}{}{W}}{T}}
     {\mathrm{bind}}{\SNV{h}{g}{L}{\Bind{}{W}{T}}}\NL
\Rule{\SNV{h}{g}{L}{U}\quad\SNV{h}{g}{L}{T}\quad
      \SCPDS{h}{g}{L}{U}{U_0}{0}\quad\SCPDS{h}{g}{L}{T}{U_0}{1}
     }{\mathrm{cast}}{\SNV{h}{g}{L}{\Cast{U}{T}}}\NL
\Rule{\SNV{h}{g}{L}{V}\quad\SNV{h}{g}{L}{T}\quad
      \SCPDS{h}{g}{L}{V}{W_0}{1}\quad\SCPDS{h}{g}{L}{T}{\Abst{}{W_0}{U_0}}{n}
     }{\mathrm{appl}}{\SNV{h}{g}{L}{\Appl{V}{T}}}
\end{Figure}


The refinement given next is needed to prove
the preservation \Thref{preserve_props}{}.

\begin{definition}[refinement for preservation of validity]
\ObjLabel{lsubsv}
\FigRef{lsubsv}{} defines
the relation $\LSubSV{h}{g}{L_1}{L_2}$ stating that
$L_1$ refines $L_2$ for preservation of stratified validity.
\qed
\end{definition}

\begin{Figure}{Refinement for preservation of stratified validity}{lsubsv}
\Rule{\SNV{h}{g}{L}{U}\quad\SNV{h}{g}{L}{T}\quad
      \MAll{n} n \le d \MImp \SCPES{h}{g}{L}{U}{T}{n}{n+1}
     }{\mathrm{hcast}}{\SHNV{h}{g}{L}{\Cast{U}{T}}{d}}\NL
\Rule{}{\mathrm{atom}}{\LSubSV{h}{g}{\Null}{\Null}} \quad
\Rule{\LSubSV{h}{g}{L_1}{L_2}}{\mathrm{pair}}
     {\LSubSV{h}{g}{\Pair{L_1}{}{W}}{\Pair{L_2}{}{W}}} \NL
\Rule{\LSubSV{h}{g}{L_1}{L_2}\quad
      \SHNV{h}{g}{L_1}{\Cast{W}{V}}{d}\quad\SNV{h}{g}{L_2}{W}\quad
       \DA{h}{g}{L_1}{V}{d+1}\quad\DA{h}{g}{L_2}{W}{d}
     }{\mathrm{beta}}
     {\LSubSV{h}{g}{\LDef{L_1}{}{(\Cast{W}{V})}}{\LDec{L_2}{}{W}}}
\end{Figure}

The main results on stratified validity and on its refinement
are in \SubRef{snv_props}.

\subsection{Closures}
\SubLabel{closures}

Most properties of $\LD$ are proved by structural induction,
but this proof method fails for some important results
like the confluence theorem.
In most cases a proof by induction on the ``proper subclosures''
provides for a good alternative.
The main exception is the preservation theorem.
Hereafter, a ``closure'' is an ordered pair $\CLOSURE{L}{T}$ 
where $T$ is a term closed in an environment $L$.
Intuitively, a subclosure of $\CLOSURE{L}{T}$ contains
a subterm of $T$ and a subenvironment of $L$.


The ``direct'' and ``transitive'' subclosures of $\CLOSURE{L}{T}$
are defined next.

\begin{definition}[direct subclosure]
\ObjLabel{fqu}
The relation $\FQU{L_1}{T_1}{L_2}{T_2}$ defined in \FigRef{fqu}{},
states that $\CLOSURE{L_2}{T_2}$ is a ``direct subclosure'' of $\CLOSURE{L_1}{T_1}$.

The relation $\FQUQ{L_1}{T_1}{L_2}{T_2}$ is its reflexive closure.
\qed
\end{definition}

\begin{Figure}{Direct subclosure}{fqu}
\Rule{}{\mathrm{pair}}{\FQU{L}{\Item{}{V}{T}}{L}{V}}\quad
\Rule{}{\mathrm{flat}}{\FQU{L}{\Flat{V}{T}}{L}{T}}\quad
\Rule{}{\mathrm{bind}}{\FQU{L}{\Bind{}{W}{T}}{\Pair{L}{}{W}}{T}}\NL
\Rule{}{\mathrm{lref}}{\FQU{\Pair{K}{}{W}}{\LRef{0}}{K}{W}}\quad
\Rule{\Drop{0}{m+1}{L}{K}\quad\Lift{0}{m+1}{T}{U}}{\mathrm{drop}}{\FQU{L}{U}{K}{T}}
\end{Figure}

The symbol $^?$ in $\FQUQ{L_1}{T_1}{L_2}{T_2}$
means ``one or none'' as for regular expressions.


\begin{definition}[subclosure and proper subclosure]
\ObjLabel{fqus}
\FigRef{fqus}{} defines 
the relation $\FQUS{L_1}{T_1}{L_2}{T_2}$ (subclosure)
as the (reflexive and) transitive closure of $\FQUQ{L_1}{T_1}{L_2}{T_2}$.
While the proper subclosure is
the transitive closure of $\FQU{L_1}{T_1}{L_2}{T_2}$.
\qed
\end{definition}

\begin{Figure}{Subclosure}{fqus}
\Rule{\FQUQ{L_1}{T_1}{L_2}{T_2}}{\mathrm{inj}}{\FQUS{L_1}{T_1}{L_2}{T_2}} \quad
\Rule{\FQUS{L_1}{T_1}{L}{T}\quad\FQUQ{L}{T}{L_2}{T_2}}{\mathrm{step}}
     {\FQUS{L_1}{T_1}{L_2}{T_2}}
\end{Figure}

We want to remark that generalizing the constant $0$
in \FigRef{fqu}{drop}, invalidates the commutation property
between the direct subclosure and the parallel reduction,
which is crucial for the preservation theorem.
Moreover the proper subclosure is well founded, as we see
by observing that each step of direct subclosure
decreases the sum of the term constructors in the closure.

\subsection{Extended Reduction}
\SubLabel{cpx_defs}

Having introduced subclosures,
we can take a glance at the strong normalization of ``$rst$-reduction''
\cite{SPAc4}, informally known as the ``big tree'' theorem.

Ideally, given a closure $\CLOSURE{L_1}{T_1}$ we define
a step $\step{r}$ along the axis of reducts,
a step $\step{s}$ along the axis of subclosures,
and a step $\step{t}$ along the axis of iterated static types:
\[
\Rule{\CPR{L_1}{T_1}{T_2}\quad T_1\neq T_2}{}
{\CLOSURE{L_1}{T_1}\step{r}\CLOSURE{L_1}{T_2}}\quad
\Rule{\FQU{L_1}{T_1}{L_2}{T_2}}{}
{\CLOSURE{L_1}{T_1}\step{s}\CLOSURE{L_2}{T_2}}\quad
\Rule{\LSTAS{h}{L_1}{T_1}{T_2}{1}\quad\DA{h}{g}{L_1}{T_1}{d+1}}{}
{\CLOSURE{L_1}{T_1}\step{t}\CLOSURE{L_1}{T_2}}
\EqLabel{rst}\]
and we are interested in proving that
any sequence of such steps staring from $\CLOSURE{L}{T}$,
is finite when $\SNV{h}{g}{L}{T}$.
This is the strong normalization of a relation $\step{rst}$
comprising the steps in \EqRef{rst}.
We remark that the interest in this result lies on
the very powerful induction principle
it provides for proving properties of valid terms.
We shall need this power for the preservation theorem.
Notice the side condition $T_1\neq T_2$
ensuring that $\step{r}$ is not reflexive
(we can prove that \DefRef{cpr}{} forbids
single-step reduction cycles),
and the side condition $\DA{h}{g}{L_1}{T_1}{d+1}$
ensuring that $\step{t}$ cannot be applied indefinitely
(otherwise,
$\CLOSURE{L}{\Sort{k}}\step{t}\CLOSURE{L}{\Sort{(\Next{h}{k})}}$
is always possible).
 
As to the proof of the ``big tree'' theorem,
we take a sequence of steps starting from a valid closure
and we would like to commute adjacent steps
until the steps of the same kind are clustered.
At that point, an infinite sequence would lead to an infinite cluster,
contradicting either strong normalization of reduction (steps of kind $\step{r}$),
or well-foundedness of subclosures (steps of kind $\step{s}$), 
or else finiteness of degree in the given system of reference $\Deg{h}{g}{}{}$
(steps of kind $\step{t}$).

Unfortunately, it is a matter of fact that a step $\step{r}$ and a step $\step{t}$
may not commute.
Consider the $\beta$-redex $T_1 = \Appl{V}{\Abst{}{W_1}{\LRef{0}}}$
and its $\beta$-reductum $T_2 = \Abbr{}{(\Cast{W_1}{V})}{\LRef{0}}$.
Then the  static type of $T_1$ is $U_1 = \Appl{V}{\Abst{}{W_1}{(\Lift{0}{1}{W_1}{})}}$,
and its $\beta$-reductum is $U_0 = \Abbr{}{(\Cast{W_1}{V})}{(\Lift{0}{1}{W_1}{})}$.
Moreover, let $W_2$ be the static type of $V$,
then the static type of $T_2$ is $U_2 = \Abbr{}{(\Cast{W_1}{V})}{(\Lift{0}{1}{W_2}{})}$.
Now compare $U_0$ and $U_2$, that is: $W_1$ and $W_2$
(respectively, the ``expected'' and the ``inferred'' type of $V$).
Even assuming that $T_1$ is valid,
these terms are the same one just up to conversion.
It is an even simpler matter of fact that a step $\step{s}$ and a step $\step{t}$
may not commute.
Consider the term $T_1$ and its static type $U_1$,
take $V$ as a subterm of $T_1$ and its static type $W_2$.
Yet $W_2$ is not a subterm of $U_1$
and may just be related to $W_1$ by conversion when $T_1$ is valid.

Anyway, a step $\step{r}$ and a step $\step{s}$ commute with the help
of reduction for environments.
In fact, we can prove the ``pentagon''
(\IE a proposition on five closures connected by five relations)
of \Ruleref{cpr_fqu_conf},
in which the reduction for environments emerges in the case
$L_1 = \LDec{K}{}{V_1}$ and $T_1 = \LRef{0}$.
\[
\Rule{\FQU{L_1}{T_1}{K}{V_1}\quad\CPR{K}{V_1}{V_2}}{}
{\MEx{L_2,T_2}
 \LPR{L_1}{L_2}\MAnd\CPR{L_2}{T_1}{T_2}\MAnd\FQU{L_2}{T_2}{K}{V_2}
}
\EqLabel{cpr_fqu_conf}\]
These considerations lead us to define the ``extended reduction'' such that:

\begin{enumerate}

\item 
it extends ordinary reduction (\IE a $\step{r}$ step)
by supporting a $\step{t}$ step;

\item\ObjLabel{smooth}
it preserves strong normalization ``smoothly''
in that little effort is expected in updating the proof
that works for ordinary reduction \cite{lambdadeltaJ1a};

\item
it preserves the commutation with subclosures
in the form of \Ruleref{cpr_fqu_conf}.

\end{enumerate}

Extended reduction is our counterpart of
``$rt$-reduction'' \cite{SPAc4}.
It comprises the transitions of \DefRef{transition}{}
and the ones listed next.

\begin{definition}[extended transitions]
\ObjLabel{xtransition}
\FigRef{xtransition}{} defines the extended redexes
and their associated transitions $t$, $l$, and $e$,
which depend on a sort degree parameter $\Deg{h}{g}{}{}$
and on an environment $L$.
The transitions $t$, $l$ and $e$ respectively
replace a sort, a declared variable, and a type annotation
with their expected type.
\qed
\end{definition}

\begin{Figure}{Extended transitions}{xtransition}
\Rule{\Deg{h}{g}{k}{d+1}}{s}{\CQX{h}{g}{L}{\Sort{k}}{\Sort{(\Next{h}{k})}}} \quad
\Rule{\Drop{0}{i}{L}{\LDec{K}{}{W_1}}\quad\Lift{0}{i+1}{W_1}{W_2}}{l}
     {\CQX{h}{g}{L}{\LRef{i}}{W_2}} \quad
\Rule{}{e}{\CQX{h}{g}{L}{\Cast{U}{T}}{U}}
\end{Figure}

The transitions $t$ and $l$ provide the support for the $t$-step of \EqRef{rst},
while the transition $e$ allows the ``smooth'' update of
the strong normalization proof advocated by \Thref{}{smooth},
as we shall see.
%
%
We present extended reduction in its parallel form to extend \DefRef{cpr}{},
with respect to which we add the rules for the transitions $t$ and $e$.
Rule $\delta$ is modified as well to include the support for transition $l$. 
\DefRef{lpr}{} and \DefRef{cprs}{} are extended accordingly.
The point of extended reduction compared to static type assignment, is that
its context rules allow to compute the static type
in every subterm and not just along the ``spine''.

\begin{definition}[extended parallel reduction for terms]
\ObjLabel{cpx}
The relation $\CPX{h}{g}{L}{T_1}{T_2}$ of \FigRef{cpx}{} indicates
one step of extended parallel reduction from $T_1$ to $T_2$ in $L$.
\qed
\end{definition}

\begin{Figure}{Extended parallel reduction for terms (single step)}{cpx}
\Rule{\CPX{h}{g}{L}{W_1}{W_2}\quad\CPX{h}{g}{\Pair{L}{}{W_1}}{T_1}{T_2}}{\mathrm{bind}}
     {\CPX{h}{g}{L}{\Bind{}{W_1}{T_1}}{\Bind{}{W_2}{T_2}}}\quad
\Rule{\CPX{h}{g}{L}{V_1}{V_2}\quad\CPX{h}{g}{L}{T_1}{T_2}}{\mathrm{flat}}
     {\CPX{h}{g}{L}{\Flat{V_1}{T_1}}{\Flat{V_2}{T_2}}}\NL
\Rule{}{\mathrm{atom}}{\CPX{h}{g}{L}{\Atom{i}}{\Atom{i}}}\quad
\Rule{\Drop{0}{i}{L}{\Pair{K}{}{W_1}}\quad\CPX{h}{g}{K}{W_1}{W_2}\quad\Lift{0}{i+1}{W_2}{V_2}}{\delta}
     {\CPX{h}{g}{L}{\LRef{i}}{V_2}}\NL
\Rule{\CPX{h}{g}{L}{V_1}{V_2}\quad\CPX{h}{g}{L}{W_1}{W_2}\quad\CPX{h}{g}{\LDec{L}{}{W_1}}{T_1}{T_2}}{\beta}
     {\CPX{h}{g}{L}{\Appl{V_1}{\Abst{}{W_1}{T_1}}}{\Abbr{}{(\Cast{W_2}{V_2})}{T_2}}}\quad
\Rule{\Deg{h}{g}{k}{d+1}}{s}{\CPX{h}{g}{L}{\Sort{k}}{\Sort{(\Next{h}{k})}}}\NL
\Rule{\CPX{h}{g}{L}{V_1}{V_2}\quad\Lift{0}{1}{V_2}{W_2}\quad\CPX{h}{g}{L}{U_1}{U_2}\quad\CPX{h}{g}{\LDef{L}{}{U_1}}{T_1}{T_2}}{\theta}
     {\CPX{h}{g}{L}{\Appl{V_1}{\Abbr{}{U_1}{T_1}}}{\Abbr{}{U_2}{\Appl{W_2}{T_2}}}}\NL
\Rule{\CPX{h}{g}{\LDef{L}{}{V}}{U_1}{U_2}\quad\Lift{0}{1}{T_2}{U_2}}{\zeta}
     {\CPX{h}{g}{L}{\Abbr{}{V}{U_1}}{T_2}}\quad
\Rule{\CPX{h}{g}{L}{T_1}{T_2}}{\epsilon}{\CPX{h}{g}{L}{\Cast{U}{T_1}}{T_2}}\quad
\Rule{\CPX{h}{g}{L}{U_1}{U_2}}{e}{\CPX{h}{g}{L}{\Cast{U_1}{T}}{U_2}}
\end{Figure}


\begin{definition}[extended parallel reduction for environments]
\ObjLabel{lpx}
\FigRef{lpx}{} defines $\LPX{h}{g}{L_1}{L_2}$ indicating
one step of extended parallel reduction from $L_1$ to $L_2$.
\qed
\end{definition}

\begin{Figure}{Extended parallel reduction for environments (single step)}{lpx}
\Rule{}{\mathrm{atom}}{\LPX{h}{g}{\Null}{\Null}} \quad
\Rule{\LPX{h}{g}{L_1}{L_2}\quad\CPX{h}{g}{L_1}{W_1}{W_2}}{\mathrm{pair}}
     {\LPX{h}{g}{\Pair{L_1}{}{W_1}}{\Pair{L_2}{}{W_2}}}
\end{Figure}


\begin{definition}[extended parallel computation]
\ObjLabel{cpxs}
The relation $\CPXS{h}{g}{L}{T_1}{T_2}$ is the transitive closure of $\CPX{h}{g}{L}{T_1}{T_2}$,
while $\LPXS{h}{g}{L_1}{L_2}$ is the transitive closure of $\LPX{h}{g}{L_1}{L_2}$.
\FigRef{cpxs}{} defines $\CPXS{h}{g}{L}{T_1}{T_2}$ for reference.
$\LPXS{h}{g}{L_1}{L_2}$ is defined in the same manner.
\qed
\end{definition}

\begin{Figure}{Extended parallel computation for terms (multi-step)}{cpxs}
\Rule{\CPR{L}{T_1}{T_2}}{\mathrm{inj}}{\CPXS{h}{g}{L}{T_1}{T_2}} \quad
\Rule{\CPXS{h}{g}{L}{T_1}{T}\quad\CPX{h}{g}{L}{T}{T_2}}{\mathrm{step}}
     {\CPXS{h}{g}{L}{T_1}{T_2}}
\end{Figure}

The main results on extended reduction are in \SubRef{cpx_props}.

\subsection{Atomic Arity Assignment}
\SubLabel{aaa_defs}


Atomic arities are simple types representing the abstract syntax of
our reducibility candidates, introduced in the next \SubRef{gcp_defs},
and replace in this role the more complex ``binary arities''
used by \cite{lambdadeltaJ1a}.
Such arities are assigned to terms according to well-established rules.
The term ``atomic'' indicates that the base constructor of these arities
is not structured.

\begin{definition}[atomic arities and their assignment]
\ObjLabel{aaa}
Atomic arities are the simple types defined in \FigRef{aaa}{}.
$\Full$ is the base type, and $\Impl{B}{A}$ is the arrow type.
Moreover the relation $\AAA{L}{T}{A}$, defined in \FigRef{aaa}{} as well,
assigns the arity $A$ to $T$ in $L$.
\qed
\end{definition}

\begin{Figure}{Atomic arities and their assignment}{aaa}
\begin{tabular}{lrl}
atomic arity & $A, B$ & $\GLet \Full \GOr \Impl{B}{A}$
\end{tabular}\\[0.5pc]
\Rule{}{\mathrm{sort}}{\AAA{L}{\Sort{k}}{\Full}}\quad
\Rule{\Drop{0}{i}{L}{\Pair{K}{}{W}}\quad\AAA{K}{W}{B}}
     {\mathrm{lref}}{\AAA{L}{\LRef{i}}{B}}\quad
\Rule{\AAA{L}{V}{B}\quad\AAA{\LDef{L}{}{V}}{T}{A}}
     {\mathrm{abbr}}{\AAA{L}{\Abbr{}{V}{T}}{A}}\NL
\Rule{\AAA{L}{W}{B}\quad\AAA{\LDec{L}{}{W}}{T}{A}}
     {\mathrm{abst}}{\AAA{L}{\Abst{}{W}{T}}{\Impl{B}{A}}}\quad
\Rule{\AAA{L}{V}{B}\quad\AAA{L}{T}{\Impl{B}{A}}}
     {\mathrm{appl}}{\AAA{L}{\Appl{V}{T}}{A}}\quad
\Rule{\AAA{L}{U}{A}\quad\AAA{L}{T}{A}}
     {\mathrm{cast}}{\AAA{L}{\Cast{U}{T}}{A}}
\end{Figure}

As a type assignment, $\AAA{L}{T}{A}$ has two interpretations:
either $A$ is the simple type of the object $T$,
or $A$ is the simple type associated to the type $T$.
In this respect, consider the map $T \mapsto \Star{T}$
that turns a term of $\LD$ into a term of $\LR$
by operating the necessary $\delta\epsilon\zeta$-reductions on $T$ and 
by replacing every abstraction in $T$, say $\Abst{}{W}{}$ in the environment $K$,
with the abstraction $\Abst{}{B}{}$ such that $\AAA{K}{W}{B}$.
Moreover, extend this map to environment entries.
Then the rules of \FigRef{aaa}{} clearly show that
$\AAA{L}{T}{A}$ implies $\NTA{}{\Star{L}}{\Star{T}}{A}$ in $\LR$
(we did not prove this fact formally yet).


We need the next refinement in order to prove the preservation of atomic arity.

\begin{definition}[refinement for preservation of atomic arity]
\ObjLabel{lsuba}
\FigRef{lsuba}{} defines the relation $\LSubA{L_1}{L_2}$ stating that
$L_1$ refines $L_2$ for preservation of atomic arity.
\qed
\end{definition}

\begin{Figure}{Refinement for preservation of atomic arity}{lsuba}
\Rule{}{\mathrm{atom}}{\LSubA{\Null}{\Null}}\quad
\Rule{\LSubA{L_1}{L_2}}{\mathrm{pair}}
     {\LSubA{\Pair{L_1}{}{W}}{\Pair{L_2}{}{W}}}\quad
\Rule{\LSubA{L_1}{L_2}\quad\AAA{L_1}{\Cast{W}{V}}{B}\quad\AAA{L_2}{W}{B}}{\mathrm{beta}}
     {\LSubA{\LDef{L_1}{}{(\Cast{W}{V})}}{\LDec{L_2}{}{W}}}
\end{Figure}

Our results on atomic arity assignment and on its refinement
are in \SubRef{aaa_props}.

\subsection{Reducibility Candidates}
\SubLabel{gcp_defs}

The ``reducibility candidates'' are subsets of $\lambda$-terms
satisfying certain ``saturation'' conditions
used to establish properties of some typed $\lambda$-calculi.
In this article we use subsets of closures,
closed under the next seven conditions,
to prove that every term having an atomic arity in an environment,
is strongly normalizing with respect to extended reduction.
%
%
We start by defining the normal terms and the strongly normalizing terms.
These definitions take into account the fact that extended reduction
is reflexive and forbids single-step cycles.

\begin{definition}[normal terms and strongly normalizing terms]
\ObjLabel{cnx}
\FigRef{cnx}{} defines $\CNX{h}{g}{L}{T}$ and $\CSX{h}{g}{L}{T}$,
stating respectively that $T$ in $L$ is normal, and that $T$ in $L$ is strongly normalizing,
for extended reduction with respect to $h$ and $\Deg{h}{g}{}{}$. 
\qed
\end{definition}

\begin{Figure}{Normal terms and strongly normalizing terms for extended reduction}{cnx}
\Rule{\MAll{T_2} (\CPX{h}{g}{L}{T_1}{T_2}) \MImp (T_1 = T_2)}
     {\mathrm{cnx}}{\CNX{h}{g}{L}{T_1}}\NL
\Rule{\MAll{T_2} (\CPX{h}{g}{L}{T_1}{T_2}) \MImp (T_1 \neq T_2) \MImp (\CSX{h}{g}{L}{T_2})}
     {\mathrm{csx}}{\CSX{h}{g}{L}{T_1}}
\end{Figure}

Notice that $\CSX{h}{g}{L}{T_1}$ is inductively defined
with base case $\CNX{h}{g}{L}{T_1}$.
In fact, $\CNX{h}{g}{L}{T_1}$ implies $\CSX{h}{g}{L}{T_1}$
since $\CSX{h}{g}{L}{T_2}$ holds by ``ex falso quodlibet''.


Given a property $\R$ on closures,
a reducibility candidate $\C$ for $\R$ 
is a subset of closures satisfying $\R$, that we describe
constructively as a relation.
So we may write $\MAppl{\C}{T,L}$ for $\CLOSURE{L}{T} \in \C$.
Our reducibility theorem states that if $\R$ is a reducibility candidate,
then every closure with an atomic arity belongs to some $\C$
and therefore, satisfies $\R$.
In formal words we can prove that $\AAA{L}{T}{A}$ implies $\MAppl{\R}{T,L}$.
Strong normalization follows from choosing
$\CSX{h}{g}{L}{T}$ as $\MAppl{\R}{T,L}$.

We are going to present Tait-style reducibility candidates \cite{Tai75},
which differ from the Girard-style reducibility candidates \cite{GTL89}
used by \cite{lambdadeltaJ1a}, in that
condition ``CR2'' is not required
(\IE $\InGCR{L}{T_1}{\C}$ and $\CPX{h}{g}{L}{T_1}{T_2}$
imply $\InGCR{L}{T_2}{\C}$), and notably, in that 
closures without an arity are allowed in $\C$. 
This simplification gives us more freedom for constructing elements of $\C$.

\begin{definition}[reducibility candidate]
\ObjLabel{gcr}
Given a subset $\R$ of closures satisfying
\FigRef{}{S} and \FigRef{}{S0} of \FigRef{gcr}{},
a reducibility candidate $\C$ for $\R$ is a subset of closures
satisfying \FigRef{}{S1} to \FigRef{}{S7} of \FigRef{gcr}{}.
The notation ``$\InGCRV{L}{\vect{V}}{\R}$''
means ``$\InGCR{L}{V}{\R}$ for each component $V$ of $\vect{V}$''.
\qed
\end{definition}

\begin{Figure}{Reducibility candidate}{gcr}
\Rule{\InGCR{L}{\Appl{\Sort{k}}{T}}{\R}}
     {\mathrm{S}}{\InGCR{L}{T}{\R}}\quad
\Rule{\Drop{l}{m}{L}{K}\quad\InGCR{K}{T}{\R}\quad\Lift{l}{m}{T}{U}}
     {\mathrm{S0}}{\InGCR{L}{U}{\R}}\NL
\Rule{\InGCR{L}{T}{\C}}
     {\mathrm{S1}}{\InGCR{L}{T}{\R}}\quad
\Rule{\InGCRV{L}{\vect{V}}{\R}\quad\Simple{T}\quad\CNX{h}{g}{L}{T}}
     {\mathrm{S2}}{\InGCR{L}{\ApplV{\vect{V}}{}{T}}{\C}}\quad
\Rule{\InGCRV{L}{\vect{V}}{\R}}
     {\mathrm{S4}}{\InGCR{L}{\ApplV{\vect{V}}{}{\Sort{k}}}{\C}}\NL
\Rule{\InGCR{L}{\ApplV{\vect{V}}{}{\Abbr{}{(\Cast{W}{V})}{T}}}{\C}}
     {\mathrm{S3}}{\InGCR{L}{\ApplV{\vect{V}}{}{\Appl{V}{\Abst{}{W}{T}}}}{\C}}\quad
\Rule{\Drop{0}{i}{L}{\Pair{K}{}{W_1}}\quad\Lift{0}{i+1}{W_1}{W_2}\quad\InGCR{L}{\ApplV{\vect{V}}{}{W_2}}{\C}}
     {\mathrm{S5}}{\InGCR{L}{\ApplV{\vect{V}}{}{\LRef{i}}}{\C}}\NL
\Rule{\InGCR{L}{V}{\R}\quad\LiftV{0}{1}{\vect{V_1}}{\vect{V_2}}\quad\InGCR{\LDef{L}{}{V}}{\ApplV{\vect{V_2}}{}{T}}{\C}}
     {\mathrm{S6}}{\InGCR{L}{\ApplV{\vect{V_1}}{}{\Abbr{}{V}{T}}}{\C}}\quad
\Rule{\InGCR{L}{\ApplV{\vect{V}}{}{U}}{\C}\quad\InGCR{L}{\ApplV{\vect{V}}{}{T}}{\C}}
     {\mathrm{S7}}{\InGCR{L}{\ApplV{\vect{V}}{}{\Cast{U}{T}}}{\C}} 
\end{Figure}


Compound reducibility candidated are built through well-established constructions.
For now we are interested just in the ``functional'' construction introduced next.

\begin{definition}[function subset]
\ObjLabel{cfun}
If $\C_1$ and $\C_2$ are subsets of closures,
then the subset $\CFun{\C_1}{\C_2}$ is defined in \FigRef{cfun}{}.
\qed
\end{definition}

\begin{Figure}{Function subset}{cfun}
\Rule{\MAll{L,W,U,\vect{c}}\Drops{\vect{c}}{L}{K}\MImp\Lifts{\vect{c}}{T}{U}
      \MImp\InGCR{L}{W}{\C_1}\MImp\InGCR{L}{\Appl{W}{U}}{\C_2}}
     {\mathrm{cfun}}{\InGCR{K}{T}{\CFun{\C_1}{\C_2}}}
\end{Figure}

Notice that the environment $L$ of $W$ possibly extends the environment $K$ of $T$.
as is required to prove \FigRef{gcr}{S6},
in which $L$ and $\LDef{L}{}{V}$ have diffrent length.


\begin{definition}[interpretation of an atomic arity]
\ObjLabel{acr}
For a subset of closures $\R$,
the subset of closures $\ACR{\R}{}{A}$
associated to the atomic arity $A$,
is defined in \FigRef{acr}{}.
\qed
\end{definition}

\begin{Figure}{Interpretation of an atomic arity as a subset of closures}{acr}
$\ACR{\R}{}{\Full} = \R$ \quad
$\ACR{\R}{}{\CFun{B}{A}} = \Impl{\ACR{\R}{}{B}}{\ACR{\R}{}{A}}$
\end{Figure}


The refinement given next is needed to state
the general form of the reducibility theorem.
In particular it expresses in $\LD$ a simultaneus substitution
like the one occurring in the reducibility theorem for $\SF$,
which is stated using the ``parametric'' reducibility of \cite{GTL89}.

\begin{definition}[refinement for reducibility]
\ObjLabel{lsubc}
The relation $\LSubC{\R}{}{L_1}{L_2}$
defined in \FigRef{lsubc}{}, states that
$L_1$ refines $L_2$ for reducibility.
\qed
\end{definition}

\begin{Figure}{Refinement for reducibility}{lsubc}
\Rule{}{\mathrm{atom}}{\LSubC{\R}{}{\Null}{\Null}}\quad
\Rule{\LSubC{\R}{}{L_1}{L_2}}{\mathrm{pair}}
     {\LSubC{\R}{}{\Pair{L_1}{}{W}}{\Pair{L_2}{}{W}}}\NL
\Rule{\LSubC{\R}{}{L_1}{L_2}\quad\InGCR{L_1}{W}{\ACR{\R}{}{B}}\quad
      \InGCR{L_1}{V}{\ACR{\R}{}{B}}\quad\AAA{L_2}{W}{B}}{\mathrm{beta}}
     {\LSubC{\R}{}{\LDef{L_1}{}{(\Cast{W}{V})}}{\LDec{L_2}{}{W}}}
\end{Figure}

The main results on candidates and on their refinement
are in \SubRef{gcp_props}.

\subsection{Lazy Equivalence}
\SubLabel{lleq_defs}

In \SubRef{gcp_defs} we defined the normalization
of a term $T$ in the environment $L$
that, by \Thref{gcp_props}{aaa_csx}, is implied by $\AAA{L}{T}{A}$.
Now we would like to define the normalization of an environment $L$
in such a way that $\AAA{L}{T}{A}$ implies it as well.
However, we notice from \FigRef{aaa}{lref} that $\AAA{L}{T}{A}$ constrains just
the entries of $L$ hereditarily referred by $T$.
Thus, following the paradigm of \FigRef{cnx}{csx},
we need to replace $T_1 \neq T_2$
with the negated equivalence $\NLLEq{}{T}{L_1}{L_2}$ 
stating that $L_1$ and $L_2$ differ in one entry hereditarily referred by $T$.
The corresponding equivalence is defined next.
Working under the assumption that every entry of $L$ has an arity,
simplifies the development significantly,
but we aim at showing that this assumption is redundant.


\begin{definition}[lazy equivalence for environments]
\ObjLabel{lleq}
The relation $\LLEq{l}{T}{L_1}{L_2}$
defined in \FigRef{lleq}{},
states that the environments $L_1$ and $L_2$ are equal
in the entries hereditarily referred by the term $T$ at level $l$.
\qed
\end{definition}

\begin{Figure}{Lazy equivalence for environments}{lleq}
\Rule{\Length{L_1} = \Length{L_2}}
     {\mathrm{sort}}{\LLEq{l}{\Sort{k}}{L_1}{L_2}} \quad
\Rule{\Length{L_1} = \Length{L_2} \quad i < l}
     {\mathrm{skip}}{\LLEq{l}{\LRef{i}}{L_1}{L_2}} \quad
\Rule{\Length{L_1} = \Length{L_2} \quad \Length{L_1} \le i \quad \Length{L_2} \le i}
     {\mathrm{free}}{\LLEq{l}{\LRef{i}}{L_1}{L_2}} \NL
\Rule{l \le i \quad
      \Drop{0}{i}{L_1}{\Pair{K_1}{}{W}} \quad \Drop{0}{i}{L_2}{\Pair{K_2}{}{W}} \quad
      \LLEq{0}{W}{K_1}{K_2}
     }{\mathrm{lref}}{\LLEq{l}{\LRef{i}}{L_1}{L_2}} \NL
\Rule{\LLEq{l}{W}{L_1}{L_2} \quad\LLEq{l+1}{T}{\Pair{L_1}{}{W}}{\Pair{L_2}{}{W}}}
     {\mathrm{bind}}{\LLEq{l}{\Bind{}{W}{T}}{L_1}{L_2}} \quad
\Rule{\LLEq{l}{V}{L_1}{L_2} \quad\LLEq{l}{T}{L_1}{L_2}}
     {\mathrm{flat}}{\LLEq{l}{\Flat{V}{T}}{L_1}{L_2}}
\end{Figure}

This relation is an equivalence that we term ``lazy''
since we check for equality just the entries of $L_1$ and $L_2$
hereditarily referred by $T$.
%
%
Its nonrecursive definition \EqRef{lleq_alt} uses ``hereditarily free'' variables.
We say that a variable is ``hereditarily free'' in $\CLOSURE{L}{T}$
when it is free in $T$ or in an entry of $L$ hereditarily referred by $T$.
This idea is expressed formally by the next definition.
Alternatively, we can say that a variable is hereditarily free in $\CLOSURE{L}{T}$
when it is free in a $\delta l$-reduct of $T$ in $L$
(see \DefRef{transition}{} and \DefRef{xtransition}{}
for $\delta$-reducts and $l$-reducts respectively).

\begin{definition}[hereditarily free variables]
\ObjLabel{frees}
\FigRef{frees}{} defines $\FreeS{i}{l}{L}{T}$,
stating that the variable introduced at depth $i$ is hereditarily free 
at level $l$ in $\CLOSURE{L}{T}$.
\qed
\end{definition}

\begin{Figure}{Hereditarily free variables}{frees}
\Rule{\MAll{T} \NLift{i}{1}{T}{U}}
     {\mathrm{eq}}{\FreeS{i}{l}{L}{U}}\NL
\Rule{l \le j \quad j < i \quad (\MAll{T} \NLift{j}{1}{T}{U}) \quad
      \Drop{0}{j}{L}{\Pair{K}{}{W}} \quad \FreeS{i-j-1}{0}{K}{W}
     }{\mathrm{be}}{\FreeS{i}{l}{L}{U}}
\end{Figure}

We need the level $l$ to reason about hereditarily free variables
in the scope of binders.

For example we can prove that
$\FreeS{i+1}{l+1}{\Pair{L}{}{W}}{U}$ implies
$\FreeS{i}{l}{L}{\Bind{}{W}{U}}$.


An ancillary operation that we term ``pointwise union'' at level $l$
of $L_1$ and $L_2$ with respect to $T$ (notation: $\LLOr{l}{T}{L_1}{L_2}{}$), 
leads to important properties connecting lazy equivalence and
parallel reduction for environments such as \Thref{lleq_props}{lleq_lpx_trans}.
The environment $\LLOr{l}{T}{L_1}{L_2}{}$
is defined when $\Length{L_1} = \Length{L_2}$
and its $i$-th entry is taken from $L_2$
if $l \le i$ and $\FreeS{i}{l}{L_1}{T}$,
or else it is taken from $L_1$.

\begin{definition}[pointwise union]
\ObjLabel{llor}
The partial operation $\LLOr{l}{T}{L_1}{L_2}{}$ defined in \FigRef{llor}{},
constructs the ``pointwise union'' at level $l$
of $L_1$ and $L_2$ with respect to $T$.
\qed
\end{definition}

\begin{Figure}{Pointwise union of environments}{llor}
\Rule{}{\mathrm{atom}}{\LLOr{l}{U}{\Null}{\Null}{\Null}}\NL
\Rule{\LLOr{l}{U}{L_1}{L_2}{L} \quad
      \NFreeS{\Length{L_1}}{l}{\Bind{1}{W_1}{L_1}}{U}
     }{\mathrm{sn}}
     {\LLOr{l}{U}{\Bind{1}{W_1}{L_1}}{\Bind{2}{W_2}{L_2}}{\Bind{1}{W_1}{L}}}\NL
\Rule{\LLOr{l}{U}{L_1}{L_2}{L} \quad\!\! l \le \Length{L_1} \quad
      \FreeS{\Length{L_1}}{l}{\Bind{1}{W_1}{L_1}}{U}
     }{\mathrm{dx}}
     {\LLOr{l}{U}{\Bind{1}{W_1}{L_1}}{\Bind{2}{W_2}{L_2}}{\Bind{2}{W_2}{L}}}
\end{Figure}


Lazy equivalence yields environments $L$
normalizing with respect to $T$
(notation $\LSX{h}{g}{l}{T}{L}$) such that 
$\CSX{h}{g}{L}{T}$ implies $\LSX{h}{g}{l}{T}{L}$ for every level $l$.
See \Thref{lsx_props}{csx_lsx}.

\begin{definition}[strongly normalizing environments]
\ObjLabel{lsx}
\FigRef{lsx}{} defines the relation $\LSX{h}{g}{l}{T}{L}$,
stating that $L$
is strongly normalizing at level $l$ for extended reduction
with respect to the parameters $\Next{h}{}$ and $\Deg{h}{g}{}{}$,
and with respect to $T$.
\qed
\end{definition}

\begin{Figure}{Strongly normalizing environments for extended reduction}{lsx}
\Rule{\MAll{L_2} (\LPX{h}{g}{L_1}{L_2}) \MImp (\NLLEq{l}{T}{L_1}{L_2}) \MImp (\LSX{h}{g}{l}{T}{L_2})}
     {\mathrm{lsx}}{\LSX{h}{g}{l}{T}{L_1}}
\end{Figure}

Notice the common structure of \FigRef{lsx}{lsx} and \FigRef{cnx}{csx}.


An ancillary predicate on environments $\LCoSX{h}{g}{l}{L}$
is needed in \Thref{lsx_props}{lsx_cpx_trans_lcosx}.
It serves $\LSX{h}{g}{l}{T}{L}$ as, for instance,
$\LPR{L_1}{L_2}$ serves $\CPR{L}{T_1}{T_2}$
in \Thref{cpr_props}{cpr_conf_lpr}.

\begin{definition}[strongly co-normalizing environments]
\ObjLabel{lcosx}
The predicate $\LCoSX{h}{g}{l}{L}$
defined in \FigRef{lcosx}{},
states that the $L$ is ``co-normalizing'' at level $l$
with respect to $\Next{h}{}$ and $\Deg{h}{g}{}{}$.
This means that every $i$-th entry of $L$ such that $i < l$,
is strongly normalizing according to \DefRef{lsx}{}.
``Co-normalizing'' refers to ``$i < l$'' as opposed to ``$l \le i$''.
\qed
\end{definition}

\begin{Figure}{Strongly co-normalizing environments for extended reduction}{lcosx}
\Rule{}{\mathrm{atom}}{\LCoSX{h}{g}{l}{\Null}}\quad
\Rule{\LCoSX{h}{g}{0}{L}}{\mathrm{skip}}{\LCoSX{h}{g}{0}{(\Pair{L}{}{W}})}\quad
\Rule{\LCoSX{h}{g}{l}{L}\quad\LSX{h}{g}{l}{W}{L}}{\mathrm{pair}}{\LCoSX{h}{g}{l+1}{(\Pair{L}{}{W}})}
\end{Figure}

The main results on lazy equivalence, pointwise union,
and strongly normalizing environments are in \SubRef{lleq_props}.
Comparing \SubRef{relocation} with \SubRef{lleq_defs},
the reader should notice that the notions defined here 
depend just on the component $l$ of the relocation pair $\CLOSURE{l}{m}$.
In this perspective, the given definitions
are the general ones instantiated for $m=\infty$. 
We present them in this form because
the parameter $m$ turns out to be unnecessary for now.

\subsection{Very Big Trees}
\SubLabel{fpb_defs}

With the help of lazy equivalence,
we can finally define our counterpart of 
``$rst$-reduction'' \cite{SPAc4}, which we
informally introduced in \SubRef{cpx_defs}.
This counterpart is actually an extension that operates on
closures. We term it ``$qrst$-reduction'' because we add a 
``$q$-step'' of lazy equivalence.


\begin{definition}[$qrst$-reduction and $qrst$-computation]
\ObjLabel{fpbs}
The relation $\FPBQ{h}{g}{L_1}{T_1}{L_2}{T_2}$
defined in \FigRef{fpbs}{},
denotes one step of $qrst$-reduction
from the closure $\CLOSURE{L_1}{T_1}$ to the closure $\CLOSURE{L_2}{T_2}$
with respect to the parameters $\Next{h}{}$ and $\Deg{h}{g}{}{}$.
The relation $\FPBS{h}{g}{L_1}{T_1}{L_2}{T_2}$ ($qrst$-computation),
defined in \FigRef{fpbs}{} as well,
is the is the (reflexive and) transitive closure of
$\FPBQ{h}{g}{L_1}{T_1}{L_2}{T_2}$.
\qed
\end{definition}

\begin{Figure}{$qrst$-reduction and $qrst$-computation}{fpbs}
\Rule{\CPX{h}{g}{L}{T_1}{T_2}}{\mathrm{cpx}}{\FPBQ{h}{g}{L}{T_1}{L}{T_2}}\quad
\Rule{\LPX{h}{g}{L_1}{L_2}}{\mathrm{lpx}}{\FPBQ{h}{g}{L_1}{T}{L_2}{T}}\NL
\Rule{\FQUQ{L_1}{T_1}{L_2}{T_2}}{\mathrm{fquq}}{\FPBQ{h}{g}{L_1}{T_1}{L_2}{T_2}}\quad
\Rule{\LLEq{0}{T}{L_1}{L_2}}{\mathrm{lleq}}{\FPBQ{h}{g}{L_1}{T}{L_2}{T}}\NL
\Rule{\FPBQ{h}{g}{L_1}{T_1}{L_2}{T_2}}{\mathrm{inj}}{\FPBS{h}{g}{L_1}{T_1}{L_2}{T_2}}\quad
\Rule{\FPBS{h}{g}{L_1}{T_1}{L}{T}\quad\FPBQ{h}{g}{L}{T}{L_2}{T_2}}{\mathrm{step}}
     {\FPBS{h}{g}{L_1}{T_1}{L_2}{T_2}}
\end{Figure}

\FigRef{fpbs}{fquq} is the ``$s$-step'',
\FigRef{fpbs}{cpx} is the ``$rt$-step'' for terms,
\FigRef{fpbs}{lpx} is the ``$rt$-step'' for environments,
and \FigRef{fpbs}{lleq} is our new ``$q$-step''.
Because of it, our ``big'' trees are actually ``very big''
with respect to \cite{SPAc4}.
Formally, the ``very big'' tree rooted at $\CLOSURE{L}{T}$
comprises the $qrst$-computations starting at $\CLOSURE{L}{T}$. 
Our ``very big tree'' theorem states that if $T$ has an atomic arity
in $L$ (\SubRef{aaa_defs}), then the nonreflexive $rst$-steps in this tree are finite.


In order to state the theorem,
the next definition highlights the
proper (\IE nonreflexive) $rst$-steps
and the $qrst$-computations containing them.

\begin{definition}[proper $rst$-reduction and proper $qrst$-computation]
\ObjLabel{fpbg}
\FigRef{fpbg}{} defines the relation $\FPB{h}{g}{L_1}{T_1}{L_2}{T_2}$,
denoting one step of proper $rst$-reduction
from $\CLOSURE{L_1}{T_1}$ to $\CLOSURE{L_2}{T_2}$
with respect to the parameters $\Next{h}{}$ and $\Deg{h}{g}{}{}$,
and the relation $\FPBG{h}{g}{L_1}{T_1}{L_2}{T_2}$,
denoting a proper $qrst$-computation.
\qed
\end{definition}

\begin{Figure}{Proper $rst$-reduction and proper $qrst$-computation}{fpbg}
\Rule{\FQU{L_1}{T_1}{L_2}{T_2}}{\mathrm{fqu}}{\FPB{h}{g}{L_1}{T_1}{L_2}{T_2}}\quad
\Rule{\CPX{h}{g}{L}{T_1}{T_2} \quad T_1 \neq T_2}{\mathrm{cpx}}{\FPB{h}{g}{L}{T_1}{L}{T_2}}\quad
\Rule{\LPX{h}{g}{L_1}{L_2}\quad\NLLEq{0}{T}{L_1}{L_2}}{\mathrm{lpx}}{\FPB{h}{g}{L_1}{T}{L_2}{T}}\NL
\Rule{\FPB{h}{g}{L_1}{T_1}{L}{T} \quad \FPBS{h}{g}{L}{T}{L_2}{T_2}}{\mathrm{fpbg}}
     {\FPBG{h}{g}{L_1}{T_1}{L_2}{T_2}}
\end{Figure}

\Thref{fpb_props}{fpb_fpbq_alt} shows that
a step of proper $rst$-reduction is never reflexive,
but a proper $qrst$-computation may be.
Consider the term $\Appl{\Delta_{k,T}}{\Delta_{k,T}}$ where
$\Delta_{k,T} = \Abst{}{T}{\Appl{\Sort{k}}{\Appl{\LRef{0}}{\LRef{0}}}}$.
Following the example of $\Appl{\Delta_T}{\Delta_T}$ in \SubRef{cpr_defs},
we can prove
$\CPR{L}{\Appl{\Delta_{k,T}}{\Delta_{k,T}}}{\Appl{\Sort{k}}{\Appl{\Delta_{k,T}}{\Delta_{k,T}}}}$
(proper $r$-step),
and then $\FQU{L}{\Appl{\Sort{k}}{\Appl{\Delta_{k,T}}{\Delta_{k,T}}}}{L}{\Appl{\Delta_{k,T}}{\Delta_{k,T}}}$ 
($s$-step) by \FigRef{fqu}{flat}.
Moreover by \Thref{fpb_props}{lleq_fpb_trans},
starting a proper $qrst$-computation with a proper step,
is not restrictive.


Now we can define the closures whose ``very big'' tree contains
a finite number of nonreflexive $rst$-steps.
This is achieved by standard means with the next definition.

\begin{definition}[$q$-equivalence and strongly $rst$-normalizing closures]
\ObjLabel{fsb}\hspace{6pt}
\FigRef{fpbg}{} defines the relation $\FLEq{l}{L_1}{T_1}{L_2}{T_2}$ ($q$-equivalence)
that extends lazy equivalence to closures,
and the predicate $\FSB{h}{g}{L}{T}$ stating that
$\CLOSURE{L}{T}$ is strongly normalizing
for $qrst$-reduction with respect to the parameters $\Next{h}{}$ and $\Deg{h}{g}{}{}$. 
\qed
\end{definition}

\begin{Figure}{$q$-equivalence and strongly $rst$-normalizing closures}{fsb}
\Rule{\LLEq{l}{T}{L_1}{L_2}}{\mathrm{fleq}}{\FLEq{l}{L_1}{T}{L_2}{T}}\quad
\Rule{\MAll{L_2,T_2} \FPB{h}{g}{L_1}{T_1}{L_2}{T_2} \MImp \FSB{h}{g}{L_2}{T_2}}
     {\mathrm{fsb}}{\FSB{h}{g}{L_1}{T_1}}
\end{Figure}

\Thref{fpb_props}{fpb_fpbq_alt} and \Thref{fpb_props}{fpbq_inv_fpb_alt} show that
$\FPB{h}{g}{L_1}{T_1}{L_2}{T_2}$ is equivalent to
$\FPBQ{h}{g}{L_1}{T_1}{L_2}{T_2} \MAnd \NFLEq{0}{L_1}{T_1}{L_2}{T_2}$,
so we can rephrase \FigRef{fsb}{fsb} following the pattern of
\FigRef{cnx}{csx} and \FigRef{lsx}{lsx}.
Moreover $\FSB{h}{g}{L}{T}$ can be generated by \Ruleref{fsb_alt},
which is \FigRef{fsb}{fsb} with $\FPBG{h}{g}{L_1}{T_1}{L_2}{T_2}$
in place of $\FPB{h}{g}{L_1}{T_1}{L_2}{T_2}$. 
So $\CLOSURE{L}{T}$ is strongly $rst$-normalizing
iff it is strongly $qrst$-normalizing.

Our results on $qrst$-computations and $qrst$-normalization
are in \SubRef{fpb_props}.

\section{Propositions on $\LD$}
\SecLabel{propositions}

In this section we present the main properties
of reduction (\SubRef{cpr_props}),
of degree assignment (\SubRef{da_props}),
of $rt$-reduction (\SubRef{cpx_props}),
of atomic arity assignment (\SubRef{aaa_props}),
of reducibility candidates (\SubRef{gcp_props}),
of lazy equivalence (\SubRef{lleq_props}),
of $qrst$-reduction (\SubRef{fpb_props}),
and finally of stratified validity (\SubRef{snv_props})
respecting the dependences between these properties.

We aim at reaching our versions of the ``three problems'' \cite{SPA94}:
\Thref{cprs_props}{cprs_conf} (confluence of computation),
\Thref{fsb_props}{aaa_fsb} (strong $qrst$-normalization of valid terms),
and \Thref{preserve_props}{snv_cprs_lpr} (subject reduction of stratified validity).

The detailed theory of $\LD$ (1416 proofs) exists
only in the digital form of \cite{lambdadeltaV2a}.
In this article we just outline the proofs of the presented statements
by reporting on the proof strategy and on the main dependences of each proof.
Most proofs are by induction on the height of a derivation
or by cases on the last step of a derivation.
Very often both techniques are applied together. 

\AppRef{pointers} lists the pointers to the digital proofs
outlined in the article.

\subsection{Results on Reduction}
\SubLabel{cpr_props}

The relevant properties of reduction, conversion, and their refinement are listed next.

\begin{theorem}[reduction and its refinement]\hfil
\ObjLabel{cpr_props}
\begin{enumerate}

\item\ObjLabel{lsubr_trans}
(transitivity of refinement)\\
If $\LSubR{L_1}{L}$ and $\LSubR{L}{L_2}$ then $\LSubR{L_1}{L_2}$. 

\item\ObjLabel{lsubr_cpr_trans}
(transitivity of reduction for terms through refinement)\\
If $\LSubR{L_1}{L_2}$ and $\CPR{L_2}{T_1}{T_2}$ then $\CPR{L_1}{T_1}{T_2}$.

\item\ObjLabel{cpr_conf_lpr}
(confluence of reduction for terms with itself, diamond property, general form)\\
If $\CPR{L_0}{T_0}{T_1}$ and $\CPR{L_0}{T_0}{T_2}$
and $\LPR{L_0}{L_1}$ and $\LPR{L_0}{L_2}$ then
there exists $T$\\
such that $\CPR{L_1}{T_1}{T}$ and $\CPR{L_2}{T_2}{T}$.

\item\ObjLabel{lpr_conf}
(confluence of reduction for environments with itself, diamond property)\\
If $\LPR{L_0}{L_1}$ and $\LPR{L_0}{L_2}$
then there exists $L$ such that
$\LPR{L_1}{L}$ and $\LPR{L_2}{L}$.

\end{enumerate}
\end{theorem}

\begin{proof}
\Thref{}{lsubr_trans} is proved by induction on its first premise
and by cases on its second premise.
\Thref{}{lsubr_cpr_trans} is proved by induction on its second premise.
\Thref{}{cpr_conf_lpr} is proved
by induction on the proper subclosures of $\CLOSURE{L_0}{T_0}$ (\SubRef{closures})
and by cases on its four premises.
Reduction for environments emerges when considering \FigRef{cpr}{\delta}
and when a binder in the ``spine'' of $T_0$
is pushed into $L_0$ in the cases of 
\FigRef{cpr}{\mathrm{bind}}, \FigRef{cpr}{\beta}, and \FigRef{cpr}{\theta}.
Moreover, \Thref{}{lsubr_cpr_trans} and \FigRef{lsubr}{\mathrm{beta}}
are invoked when \FigRef{cpr}{\beta} is considered.
\Thref{}{lpr_conf} is proved by induction on $\Length{L_0}$
and by cases on its two premises
with the help of \Thref{}{cpr_conf_lpr}. 
\end{proof}

\begin{theorem}[computation and conversion]\hfil
\ObjLabel{cprs_props}
\begin{enumerate}

\item\ObjLabel{cprs_conf}
(confluence of computation for terms with itself, Church-Rosser property)\\
If $\CPRS{L}{T_0}{T_1}$ and $\CPRS{L}{T_0}{T_2}$
then there exists $T$\\
such that $\CPRS{L}{T_1}{T}$ and $\CPRS{L}{T_2}{T}$.

\item\ObjLabel{lprs_conf}
(confluence of computation for environments with itself, Church-Rosser property)\\
If $\LPRS{L_0}{L_1}$ and $\LPRS{L_0}{L_2}$
then there exists $L$ such that
$\LPRS{L_1}{L}$ and $\LPRS{L_2}{L}$.

\item\ObjLabel{cpcs_inv_cprs}
(formulation of conversion as a pair of confluent computations)\\
If $\CPCS{L}{T_1}{T_2}$ then there exists $T$
such that $\CPRS{L}{T_1}{T}$ and $\CPRS{L}{T_2}{T}$.

\end{enumerate}
\end{theorem}

\begin{proof}
\Thref{}{cprs_conf} and \Thref{}{lprs_conf}
are proved by induction on their first premise by invoking
the corresponding ``strip'' lemmas \cite{Brn92} from
\Thref{cpr_props}{cpr_conf_lpr} and \Thref{cpr_props}{lpr_conf} respectively. 
\Thref{}{cpcs_inv_cprs} is proved by induction on its premise
with the help of the ``strip'' lemma from \Thref{cpr_props}{cpr_conf_lpr}.
\end{proof}

The main result on reduction is Church-Rosser property,
also known as the confluence theorem and one of the so-called ``three problems''
in the Automath tradition.
The main result on conversion is
its formulation as a pair of confluent computations:
one direction is \Thref{cprs_props}{cpcs_inv_cprs},
the reverse is straightforward.
Using this formulation, 
\Thref{cpr_props}{cpr_conf_lpr} and \Thref{cprs_props}{cprs_conf}, 
give the generation lemma on abstraction,
a desired property mentioned by \cite{SPA2}.
This lemma states that $\CPCS{L}{\Abst{}{W_1}{T_1}}{\Abst{}{W_2}{T_2}}$
implies $\CPCS{L}{W_1}{W_2}$ and $\CPCS{\LDec{L}{}{W_1}}{T_1}{T_2}$.

\subsection{Results on Degree Assignment}
\SubLabel{da_props}

The relevant properties of degree assignment and of its refinement are listed next.

\medskip\begin{theorem}[degree assignment and its refinement]\hfil
\ObjLabel{da_props}
\begin{enumerate}

\item\ObjLabel{da_lstas}
(equivalence of degree assignment and iterated static type assignment, left to right)\\
If $\DA{h}{g}{L}{T}{d}$ then for each $n$ there exists $U$\\
such that $\LSTAS{h}{L}{T}{U}{n}$ and $\DA{h}{g}{L}{U}{d-n}$.

\item\ObjLabel{lstas_inv_da}
(equivalence of degree assignment and iterated static type assignment, right to left)\\
If $\LSTAS{h}{L}{T}{U}{n}$ then
for each $\Deg{h}{g}{}{}$ there exists $d$\\
such that $\DA{h}{g}{L}{T}{d}$ then $\DA{h}{g}{L}{U}{d-n}$.

\item\ObjLabel{lstas_inv_da_ge}
(equivalence of degree assignment and iterated static type assignment, variant)\\
If $\LSTAS{h}{L}{T}{U}{n}$ then
for every $n_0$ there exist $\Deg{h}{g}{}{}$ and $d \ge n_0$\\
such that $\DA{h}{g}{L}{T}{d}$ and $\DA{h}{g}{L}{U}{d-n}$.

\item\ObjLabel{lsubd_fwd_lsubr}
(inclusion of refinement)\\
If $\LSubD{h}{g}{L_1}{L_2}$ then $\LSubR{L_1}{L_2}$.

\item\ObjLabel{lsubd_da_trans}
(transitivity of degree assignment through refinement)\\
If $\LSubD{h}{g}{L_1}{L_2}$ and $\DA{h}{g}{L_2}{T}{d}$ then $\DA{h}{g}{L_1}{T}{d}$.

\item\ObjLabel{lsubd_da_conf}
(confluence of refinement and degree assignment)\\
If $\LSubD{h}{g}{L_1}{L_2}$ and $\DA{h}{g}{L_1}{T}{d}$ then $\DA{h}{g}{L_2}{T}{d}$.

\item\ObjLabel{lsubd_trans}
(transitivity of refinement)\\
If $\LSubD{h}{g}{L_1}{L}$ and $\LSubD{h}{g}{L}{L_2}$ then $\LSubD{h}{g}{L_1}{L_2}$. 

\end{enumerate}
\end{theorem}

\begin{proof}
\Thref{}{da_lstas}, \Thref{}{lstas_inv_da},
\Thref{}{lstas_inv_da_ge}, and \Thref{}{lsubd_fwd_lsubr}
are proved by induction on the premise.
\Thref{}{lsubd_da_trans} and \Thref{}{lsubd_da_conf}
are proved by induction on the second premise and by cases on the
first premise.
\Thref{}{lsubd_trans} is proved
by induction on its first premise and by cases on its second premise
by invoking \Thref{}{lsubd_da_trans} and \Thref{}{lsubd_da_conf}.
\end{proof}

\Thref{da_props}{da_lstas} and \Thref{da_props}{lstas_inv_da_ge}
are the main properties of degree assignment,
from which we derive the next \Thref{lstas_props}{lstas_inv_refl_pos}
(notice that in \cite{lambdadeltaJ1a} we were able to prove it just for $n = 0$).

\medskip\begin{theorem}[iterated static type assignment]\hfil
\ObjLabel{lstas_props}
\begin{enumerate}

\item\ObjLabel{lstas_mono}
(uniqueness of iterated static type assignment)\\
If $\LSTAS{h}{L}{T}{U_1}{n}$ and $\LSTAS{h}{L}{T}{U_2}{n}$ then $U_1 = U_2$.

\item\ObjLabel{lstas_inv_refl_pos}
(irreflexivity of static type assignment iterated at least once)\\
$\LSTAS{h}{L}{T}{T}{n+1}$ is contradictory.

\end{enumerate}
\end{theorem}

\begin{proof}
\Thref{}{lstas_mono} is proved
by induction on its first premise and
by cases on its second premise.
\Thref{}{lstas_inv_refl_pos}
is proved directly with the help of
\Thref{da_props}{lstas_inv_da_ge}.
\end{proof}

\subsection{Results on Extended Reduction}
\SubLabel{cpx_props}

The relevant properties of extended reduction are listed next.

\begin{theorem}[extended reduction]\hfil
\ObjLabel{cpx_props}
\begin{enumerate}

\item\ObjLabel{lsubr_cpx_trans}
(transitivity of extended reduction for terms through refinement)\\
If $\LSubR{L_1}{L_2}$ and $\CPX{h}{g}{L_2}{T_1}{T_2}$ then $\CPX{h}{g}{L_1}{T_1}{T_2}$.

\item\ObjLabel{cpr_cpx}
(inclusion of reduction, ``$r$-step'')\\
If $\CPR{L}{T_1}{T_2}$ then $\CPX{h}{g}{L}{T_1}{T_2}$.

\item\ObjLabel{sta_cpx}
(inclusion of static type assignment, ``$t$-step'')\\
If $\LSTAS{h}{L}{T_1}{T_2}{1}$ and $\DA{h}{g}{L}{T_1}{d+1}$ then $\CPX{h}{g}{L}{T_1}{T_2}$.

\item\ObjLabel{fqu_cpx_trans}
(commutation of direct subclosure with extended reduction for terms)\\
If $\FQU{L}{T_1}{K}{V_1}$ and $\CPX{h}{g}{K}{V_1}{V_2}$
then there exists $T_2$\\
such that $\CPX{h}{g}{L}{T_1}{T_2}$ and $\FQU{L}{T_2}{K}{V_2}$.

\item\ObjLabel{lpx_fqu_trans}
(commutation of extended reduction for environments with direct subclosure)\\
If $\LPX{h}{g}{L_1}{L_2}$ and $\FQU{L_2}{T_2}{K_2}{V}$
then there exist $K_1$ and $T$\\ such that
$\CPX{h}{g}{L_1}{T_2}{T}$ and $\FQU{L_1}{T}{K_1}{V}$ and $\LPX{h}{g}{K_1}{K_2}$.

\item\ObjLabel{lpx_cpx_trans}
(absorption of extended reduction for environments)\\
If $\LPX{h}{g}{L_1}{L_2}$ and $\CPX{h}{g}{L_2}{T_1}{T_2}$
then $\CPXS{h}{g}{L_1}{T_1}{T_2}$.

\item\ObjLabel{cpxs_fwd_beta}
(extended computation from a $\beta$-redex)\\
If $\CPXS{h}{g}{L}{\Appl{V}{\Abst{}{W}{T_1}}}{T_2}$ then
either $\TSTS{\Appl{V}{\Abst{}{W}{T_1}}}{T_2}$ or
$\CPXS{h}{g}{L}{\Abbr{}{(\Cast{W}{V})}{T_1}}{T_2}$.

\end{enumerate}
\end{theorem}

\begin{proof}
\Thref{}{lsubr_cpx_trans} is proved
by induction on its second premise and by cases on its first premise.
For the reference to a declaration, \FigRef{cpx}{\delta}, we have
$T_1 = \LRef{i}$, and $\Drop{0}{i}{L_2}{\LDec{K_2}{}{W_1}}$, and
$\CPX{h}{g}{K_2}{W_1}{W_2}$, and $\Lift{0}{i+1}{W_2}{T_2}$.  
It may be the case, not occurring with ordinary reduction, that
$\Drop{0}{i}{L_1}{\LDef{K_1}{}{(\Cast{W_1}{V_1})}}$ and $\LSubR{K_1}{K_2}$
for some $K_1$ and $V_1$ by \FigRef{lsubr}{\mathrm{beta}}. 
In that event the induction hypothesis yields $\CPX{h}{g}{K_1}{W_1}{W_2}$
and \FigRef{cpx}{e} gives $\CPX{h}{g}{K_1}{\Cast{W_1}{V_1}}{W_2}$
so \FigRef{cpx}{\delta} concludes $\CPX{h}{g}{L_1}{\LRef{i}}{T_2}$.
Here we see the purpose of $e$-reduction
and of the expected type $W_1$ in the $\beta$-reduced item
$\Abbr{}{(\Cast{W_1}{V_1})}{}$.
The untyped $\beta$-reduced item $\Abbr{}{V_1}{}$ of \cite{lambdadeltaJ1a}
shows here its weakness causing \Thref{}{cpxs_fwd_beta} to fail.
\Thref{}{cpr_cpx} is proved by induction on its premise.
\Thref{}{sta_cpx} is proved by induction on its first premise
and by cases on its second premise after replacing $\LSTAS{h}{L}{T_1}{T_2}{1}$
with $\LSTAS{h}{L}{T_1}{T_2}{n}$ and $n=1$.
\Thref{}{fqu_cpx_trans} is proved by cases on its first premise.
\Thref{}{lpx_fqu_trans} is proved by cases on its second premise
and then by cases on its first premise.
\Thref{}{lpx_cpx_trans} is proved
by induction on its second premise and by cases on its first premise.
\Thref{}{cpxs_fwd_beta} is proved
directly with the help of \Thref{}{lsubr_cpx_trans}
and \FigRef{lsubr}{\mathrm{beta}}. 
\end{proof}

The ``transitivity through refinement'',
\Thref{cpr_props}{lsubr_cpr_trans} and \Thref{cpx_props}{lsubr_cpx_trans},
is the crucial property that holds for ordinary reduction
and that extended reduction must preserve in order to guarantee
the ``smooth'' update of the strong normalization proof advocated in \SubRef{cpx_defs}.
In particular, extended reduction preserves \Thref{cpx_props}{cpxs_fwd_beta}, 
and thus preserves the saturation condition of \FigRef{gcr}{S3}
for the subset of strongly normalizing closures.
Another interesting property of extended reduction is
the ``square'' of \Thref{cpx_props}{fqu_cpx_trans},
which improves the ``pentagon'' of \Ruleref{cpr_fqu_conf}.
Notice that a transition $l$ makes the fifth ``side'' disappear.

Unfortunately, the ``pentagon'' remains in \Thref{cpx_props}{lpx_fqu_trans},
where the extended reduction for terms is needed in the case
$L_2 = \Pair{K_2}{}{V}$ and $T_2 = \LRef{0}$.

\Thref{cpx_props}{lpx_cpx_trans} shows that
extended computation for environments is generated by the next rules
that resemble \FigRef{lpx}{}. The same holds for ordinary computation.
\[
\Rule{}{\mathrm{atom}}{\LPXS{h}{g}{\Null}{\Null}} \quad
\Rule{\LPXS{h}{g}{L_1}{L_2}\quad\CPXS{h}{g}{L_1}{W_1}{W_2}}{\mathrm{pair}}
     {\LPXS{h}{g}{\Pair{L_1}{}{W_1}}{\Pair{L_2}{}{W_2}}}
\EqLabel{lpxs_alt}\]

\subsection{Results on Atomic Arity Assignment}
\SubLabel{aaa_props}

The properties of atomic arity assignment and of its refinement are listed next.

\begin{theorem}[arity assignment and its refinement]\hfil
\ObjLabel{aaa_props}
\begin{enumerate}

\item\ObjLabel{lsuba_fwd_lsubr}
(inclusion of refinement)\\
If $\LSubA{L_1}{L_2}$ then $\LSubR{L_1}{L_2}$.

\item\ObjLabel{lsuba_aaa_trans}
(transitivity of assignment through refinement)\\
If $\LSubA{L_1}{L_2}$ and $\AAA{L_2}{T}{A}$ then $\AAA{L_1}{T}{A}$.

\item\ObjLabel{lsuba_aaa_conf}
(confluence of refinement and assignment)\\
If $\LSubA{L_1}{L_2}$ and $\AAA{L_1}{T}{A}$ then $\AAA{L_2}{T}{A}$.

\item\ObjLabel{lsuba_trans}
(transitivity of refinement)\\
If $\LSubA{L_1}{L}$ and $\LSubA{L}{L_2}$ then $\LSubA{L_1}{L_2}$. 

\item\ObjLabel{aaa_mono}
(uniqueness of atomic arities)\\
If $\AAA{L}{T}{A_1}$ and $\AAA{L}{T}{A_2}$ then $A_1 = A_2$.

\item\ObjLabel{aaa_lstas}
(inclusion of assignment)\\
If $\AAA{L}{T}{A}$ then for each $h$ and $n$ then there exists $U$\\
such that $\LSTAS{h}{L}{T}{U}{n}$ and $\AAA{L}{U}{A}$.

\item\ObjLabel{cpx_lpx_aaa_conf}
(preservation of atomic arity through extended reduction, general form)\\
If $\AAA{L_1}{T_1}{A}$ and $\CPX{h}{g}{L_1}{T_1}{T_2}$
and $\LPX{h}{g}{L_1}{L_2}$ then $\AAA{L_2}{T_2}{A}$.

\end{enumerate}
\end{theorem}

\begin{proof}
\Thref{}{lsuba_fwd_lsubr} and \Thref{}{aaa_lstas}
are proved by induction on their premise.
\Thref{}{lsuba_aaa_trans} and \Thref{}{lsuba_aaa_conf}
are proved by induction on their second premise and by cases on their
first premise.
\Thref{}{lsuba_trans} is proved
by induction on its first premise and by cases on its second premise
with the help of \Thref{}{lsuba_aaa_trans} and \Thref{}{lsuba_aaa_conf}.
\Thref{}{aaa_mono} is proved
by induction on its first premise and by cases on its second premise.
\Thref{}{cpx_lpx_aaa_conf} is proved
by induction on its first premise by cases on its second premise
and then by cases on its third premise.
As for \Thref{cpr_props}{cpr_conf_lpr}, 
reduction for environments emerges when considering \FigRef{cpx}{\delta}
and when a binder in the ``spine'' of $T_1$
is pushed into $L_1$ in the cases of 
\FigRef{cpx}{\mathrm{bind}}, \FigRef{cpx}{\beta}, and \FigRef{cpx}{\theta}.
Moreover, \Thref{}{lsuba_aaa_trans} and \FigRef{lsuba}{\mathrm{beta}}
are invoked when \FigRef{cpx}{\beta} is considered.
\end{proof}

\Thref{aaa_props}{cpx_lpx_aaa_conf} (proposition 500 of \cite{lambdadeltaV2a})
states the ``subject reduction'' property of the arity assignment,
a prerequisite for the preservation \Thref{preserve_props}{}.

\subsection{Results on Reducibility Candidates}
\SubLabel{gcp_props}

The properties of reducibility candidates and of their refinement are listed next.

\begin{theorem}[reducibility candidates and their refinement]\hfil
\ObjLabel{gcp_props}
\begin{enumerate}

\item\ObjLabel{csx_gcr}
(the candidate of strongly normalizing closures for extended reduction)\\
For any $h$ and $\Deg{h}{g}{}{}$, the subset $\css{\CLOSURE{L}{T} \st \CSX{h}{g}{L}{T}}$\\
is a reducibility candidate for itself.

\item\ObjLabel{acr_gcr}
(the candidate associated to an atomic arity)\\
If $\R$ is a reducibility candidate for itself\\
then $\ACR{\R}{}{A}$ is a reducibility candidate for $\R$.

\item\ObjLabel{acr_aaa_csubc_lifts}
(reducibility theorem for extended reduction, general form)\\
If $\R$ is a reducibility candidate for itself then\\
$\LSubC{\R}{}{L_1}{L_2}$ and
$\Drops{\vect{c}}{L_2}{K_2}$ and
$\AAA{K_2}{T}{A}$ and
$\Lifts{\vect{c}}{T}{U}$
imply $\InGCR{L_1}{U}{\ACR{\R}{}{A}}$.

\item\ObjLabel{gcr_aaa}
(reducibility theorem for extended reduction)\\
If $\R$ is a reducibility candidate for itself then
$\AAA{L}{T}{A}$ implies $\InGCR{L}{T}{\R}$.

\item\ObjLabel{aaa_csx}
(strong normalization theorem for extended reduction)\\
If $\AAA{L}{T}{A}$ then $\CSX{h}{g}{L}{T}$.

\item\ObjLabel{lsubc_fwd_lsubr}
(inclusion of refinement)\\
If $\LSubC{\R}{}{L_1}{L_2}$ then $\LSubR{L_1}{L_2}$.

\item\ObjLabel{lsuba_lsubc}
(inverse inclusion of refinement)\\
If $\R$ is a reducibility candidate for itself then
$\LSubA{L_1}{L_2}$ implies $\LSubC{\R}{}{L_1}{L_2}$.

\end{enumerate}
\end{theorem}

\begin{proof}
\Thref{}{csx_gcr} is proved directly
by invoking \Thref{cpx_props}{cpxs_fwd_beta}
and similar propositions (one for each extended redex).
\Thref{}{acr_gcr} is proved by induction on $A$.
\Thref{}{acr_aaa_csubc_lifts} is proved
by induction on $\AAA{K_2}{T}{A}$ and by cases on the other premises
by invoking \Thref{}{acr_gcr}.
Multiple relocation emerges from \FigRef{cfun}{\mathrm{cfun}},
while the refinement emerges since \FigRef{lsubc}{\mathrm{beta}}
is needed when $T$ is a $\lambda$-abstraction.
\Thref{aaa_props}{aaa_mono} is invoked when $T$ is a reference to a declaration
in the case of \FigRef{lsubc}{\mathrm{beta}}.
\Thref{}{gcr_aaa} is a corollary of \Thref{}{acr_aaa_csubc_lifts}
and of \FigRef{gcr}{\mathrm{S1}}.
\Thref{}{aaa_csx} is a corollary of \Thref{}{gcr_aaa}
and of \Thref{}{csx_gcr}.
\Thref{}{lsubc_fwd_lsubr} is proved by induction on its premise.
\Thref{}{lsuba_lsubc} is proved by induction on its premise
with the help of \Thref{}{acr_aaa_csubc_lifts}.
\end{proof}

\Thref{gcp_props}{csx_gcr} is the most relevant property of
strongly normalizing terms.
Moreover the relation $\CSX{h}{g}{L}{T}$ is generated by the next rule
resembling \FigRef{cnx}{csx}.
\[
\Rule{\MAll{T_2} (\CPXS{h}{g}{L}{T_1}{T_2}) \MImp (T_1 \neq T_2) \MImp (\CSX{h}{g}{L}{T_2})}
     {\mathrm{csx}}{\CSX{h}{g}{L}{T_1}}
\EqLabel{csx_alt}\]

\subsection{Results on Lazy Equivalence}
\SubLabel{lleq_props}

The relevant properties of pointwise union and lazy equivalence are listed next.

We give alternative definitions of lazy equivalence.
The nonrecursive definition \EqRef{lleq_alt}
is more appropriate for the proofs we shall present.
A nonrecursive definition of pointwise union
in the style of \EqRef{lleq_alt} is available as well.
It is not easy to read, though.
\begin{small}\[
\Rule{\begin{Premise}{.}{.}\\\\
      \Length{L_1} = \Length{L_2}
      \end{Premise}\quad
      \begin{Premise}{(}{)}
      \MAll{K_1,K_2,W_1,W_2,i}\\
      l \le i \MImp \FreeS{i}{l}{L_1}{T} \MImp
      \Drop{0}{i}{L_1}{\Pair{K_1}{1}{W_1}} \MImp
      \Drop{0}{i}{L_2}{\Pair{K_2}{2}{W_2}} \MImp \\
      \BIND{1} = \BIND{2} \MAnd W_1 = W_2 
      \end{Premise}
     }{\mathrm{lleq}}
     {\LLEq{l}{T}{L_1}{L_2}}
\EqLabel{lleq_alt}\]\end{small}
\begin{small}\[
\Rule{\begin{Premise}{.}{.}\\\\
      \Length{L_1} = \Length{L_2}
      \end{Premise}\quad
      \begin{Premise}{(}{)}
      \MAll{K_1,K_2,W_1,W_2,i}\\
      l \le i \MImp (\MAll{U} \NLift{i}{1}{U}{T}) \MImp
      \Drop{0}{i}{L_1}{\Pair{K_1}{1}{W_1}} \MImp
      \Drop{0}{i}{L_2}{\Pair{K_2}{2}{W_2}} \MImp \\
      \BIND{1} = \BIND{2} \MAnd W_1 = W_2 \MAnd \LLEq{0}{W_1}{K_1}{K_2}
      \end{Premise}
     }{\mathrm{lleq}}
     {\LLEq{l}{T}{L_1}{L_2}}
\EqLabel{lleq_alt_rec}\]\end{small}

\begin{theorem}[pointwise union]\hfil
\ObjLabel{llor_props}
\begin{enumerate}

\item\ObjLabel{llor_tail_frees}
(construction lemma for tail binder, positive case)\\
If $\FreeS{\Length{L_1}}{l}{\Bind{1}{W_1}{L_1}}{U}$
and $l \le \Length{L_1}$ then\\ 
$\LLOr{l}{U}{L_1}{L_2}{L}$ implies
$\LLOr{l}{U}{(\Bind{1}{W_1}{L_1})}{(\Bind{2}{W_2}{L_2})}{\Bind{2}{W_2}{L}}$.

\item\ObjLabel{llor_tail_cofrees}
(construction lemma for tail binder, negative case)\\
If $\NFreeS{\Length{L_1}}{l}{\Bind{1}{W_1}{L_1}}{U}$ then\\ 
$\LLOr{l}{U}{L_1}{L_2}{L}$ implies
$\LLOr{l}{U}{(\Bind{1}{W_1}{L_1})}{(\Bind{2}{W_2}{L_2})}{\Bind{1}{W_1}{L}}$.

\item\ObjLabel{llor_total}
(existence lemma)\\
If $\Length{L_1} = \Length{L_2}$ then for each $T$ and $l$
then there exists $L$ such that $\LLOr{l}{T}{L_1}{L_2}{L}$.

\end{enumerate}
\end{theorem}

\begin{proof}
\Thref{}{llor_tail_frees} and \Thref{}{llor_tail_cofrees}
are proved by cases on their last premise.
\Thref{}{llor_total} is proved
by induction on $\Length{L_1}$ with the help of
\Thref{}{llor_tail_frees} and \Thref{}{llor_tail_cofrees}.
\end{proof}

\Thref{llor_props}{llor_total} (proposition 1400 of \cite{lambdadeltaV2a})
needs tail binders (\DefRef{append}{}). 

\begin{theorem}[lazy equivalence]\hfil
\ObjLabel{lleq_props}
\begin{enumerate}

\item\ObjLabel{llpx_sn_llor_fwd_sn}
(left operand lemma)\\
If $\LLEq{l}{T}{L_1}{L_2}$ and $\LPX{h}{g}{L_2}{K_2}$
and $\LLOr{l}{T}{L_1}{K_2}{K_1}$
then $\LPX{h}{g}{L_1}{K_1}$.

\item\ObjLabel{llpx_sn_llor_dx}
(right operand lemma)\\
If $\LLEq{l}{T}{L_1}{L_2}$ and $\LPX{h}{g}{L_2}{K_2}$
and $\LLOr{l}{T}{L_1}{K_2}{K_1}$
then $\LLEq{l}{T}{K_2}{K_1}$.

\item\ObjLabel{lleq_lreq_trans}
(transitivity with ranged equivalence)\\
If $\LLEq{l}{T}{L_1}{L}$ and $(\MAll{m} \LREq{l}{m}{L}{L_2})$
then $\LLEq{l}{T}{L_1}{L_2}$.

\item\ObjLabel{lleq_fqu_trans}
(transitivity with direct subclosure)\\
If $\LLEq{0}{T}{L_1}{L_2}$ and $\FQU{L_2}{T}{K_2}{U}$
then there exists $K_1$\\
such that $\FQU{L_1}{T}{K_1}{U}$ and $\LLEq{0}{U}{K_1}{K_2}$.

\item\ObjLabel{lleq_cpx_trans}
(transitivity with extended reduction for terms)\\
If $\LLEq{0}{T_1}{L_1}{L_2}$ and $\CPX{h}{g}{L_2}{T_1}{T_2}$
then $\CPX{h}{g}{L_1}{T_1}{T_2}$.

\item\ObjLabel{lleq_lpx_trans}
(transitivity with extended reduction for environments)\\
If $\LLEq{l}{T}{L_1}{L_2}$ and $\LPX{h}{g}{L_2}{K_2}$
then there exists $K_1$\\
such that $\LPX{h}{g}{L_1}{K_1}$ and $\LLEq{l}{T}{K_1}{K_2}$.

\item\ObjLabel{cpx_lleq_conf_sn}
(confluence with extended reduction for terms)\\
If $\CPX{h}{g}{L_1}{T_1}{T_2}$ then 
$\LLEq{0}{T_1}{L_1}{L_2}$ implies $\LLEq{0}{T_2}{L_1}{L_2}$.

\end{enumerate}
\end{theorem}

\begin{proof}
\Thref{}{llpx_sn_llor_fwd_sn} and \Thref{}{llpx_sn_llor_dx}
are proved directly by accessing to lazy equivalence through \Ruleref{lleq_alt}.
\Thref{}{lleq_lreq_trans} is proved
by induction on its first premise.
\Thref{}{lleq_fqu_trans} and \Thref{}{lleq_cpx_trans}
are proved by induction on their second premise and by cases on their first premise.
\Thref{}{lleq_lpx_trans} follows from
\Thref{}{llpx_sn_llor_fwd_sn} and \Thref{}{llpx_sn_llor_dx}
by taking $K_1 = \LLOr{l}{T}{L_1}{K_2}{}$,
which results from \Thref{llor_props}{llor_total}.
Here we see the purpose of pointwise union.
\Thref{}{cpx_lleq_conf_sn} is proved
by induction on its first premise and by cases on its second premise
with the help of \Thref{}{lleq_lreq_trans} when
\FigRef{cpx}{\mathrm{bind}}, \FigRef{cpx}{\beta}, and \FigRef{cpx}{\theta}
are considered. 
Here we see the purpose of ranged equivalence.
\end{proof}

The shape of the second premise in \Thref{lleq_props}{lleq_lreq_trans}
is due the implicit instantiation of $m$ with $\infty$ in \DefRef{lleq}{}
(lazy equivalence) as noted at the end of \SubRef{lleq_defs}.
\Thref{lleq_props}{lleq_lpx_trans} and \Thref{lleq_props}{cpx_lleq_conf_sn}
(proposition 1000 of \cite{lambdadeltaV2a})
are the most interesting properties of lazy equivalence
with respect to extended reduction.
Their proofs were the most demanding of this set.

\begin{theorem}[strongly normalizing environments]\hfil
\ObjLabel{lsx_props}
\begin{enumerate}

\item\ObjLabel{lsx_cpx_trans_lcosx}
(transitivity of strong normalization for environments through extended reduction)\\
If $\LCoSX{h}{g}{l}{L}$ and $\CPX{h}{g}{L}{T_1}{T_2}$,
then $\LSX{g}{g}{l}{T_1}{L}$ implies $\LSX{g}{g}{l}{T_2}{L}$.

\item\ObjLabel{lsx_lref_be_lpxs}
(construction lemma for variable reference, general form)\\
If $l \le i$ and $\CSX{h}{g}{K_1}{W}$ and $\LPXS{h}{g}{K_1}{K_2}$,\\
then $\Drop{0}{i}{L_2}{\Pair{K_2}{}{W}}$ and $\LSX{h}{g}{0}{W}{K_2}$
imply $\LSX{h}{g}{l}{\LRef{i}}{L_2}$.

\item\ObjLabel{csx_lsx}
(strong normalization for terms implies strong normalization for environments)\\
If $\CSX{h}{g}{L}{T}$ then $\LSX{h}{g}{l}{T}{L}$ for every $l$.

\end{enumerate}
\end{theorem}

\begin{proof}
\Thref{}{lsx_cpx_trans_lcosx} is proved
by induction on its second premise and by cases on its third premise.
Strongly co-normalizing environments (\DefRef{lcosx}{}) emerge when 
$T_1 = \LRef{i}$ with $i < l$ and \FigRef{cpx}{\delta} is considered.
Every construction lemma is needed except for 
\Thref{}{lsx_lref_be_lpxs}, which is proved
by induction on $\CSX{h}{g}{K_1}{W}$ using \Ruleref{csx_alt} and
then by induction on $\LSX{h}{g}{0}{W}{K_2}$
with the help of \Thref{}{lsx_cpx_trans_lcosx}
and of \Thref{cpx_props}{lpx_cpx_trans}.
\Thref{}{csx_lsx} is proved by induction
on the proper subclosures of $\CLOSURE{L}{T}$
with the help of every construction lemma
including \Thref{}{lsx_lref_be_lpxs}.
\end{proof}

\Thref{lsx_props}{csx_lsx} is the most relevant property of
strongly normalizing environments.

Notice that $\LSX{h}{g}{l}{T}{L}$ is generated by the next rule
resembling \FigRef{lsx}{lsx}.
\[
\Rule{\MAll{L_2} (\LPXS{h}{g}{L_1}{L_2}) \MImp (\NLLEq{l}{T}{L_1}{L_2}) \MImp (\LSX{h}{g}{l}{T}{L_2})}
     {\mathrm{lsx}}{\LSX{h}{g}{l}{T}{L_1}}
\EqLabel{lsx_alt}\]

\subsection{Results on Very Big Trees}
\SubLabel{fpb_props}

The properties of $qrst$-computations and strong $qrst$-normalization
are listed next.

\begin{theorem}[$qrst$-computations]\hfil
\ObjLabel{fpb_props}
\begin{enumerate}

\item\ObjLabel{fpbs_inv_alt}
(decomposition property for $qrst$-computation)\\
If $\FPBS{h}{g}{L_1}{T_1}{L_2}{T_2}$ then
there exist $L_0$, $L$, and $T$ such that\\
$\CPXS{h}{g}{L_1}{T_1}{T}$ and $\FQUS{L_1}{T}{L_0}{T_2}$ and
$\LPXS{h}{g}{L_0}{L}$ and $\LLEq{0}{T_2}{L}{L_2}$.

\item\ObjLabel{fpb_fpbq_alt}
(formulation of proper $rst$-reduction with $q$-equivalence, left to right)\\
If $\FPB{h}{g}{L_1}{T_1}{L_2}{T_2}$ then
$\FPBQ{h}{g}{L_1}{T_1}{L_2}{T_2}$ and $\NFLEq{0}{L_1}{T_1}{L_2}{T_2}$.

\item\ObjLabel{fpbq_inv_fpb_alt}
(formulation of proper $rst$-reduction with $q$-equivalence, right to left)\\
If $\FPBQ{h}{g}{L_1}{T_1}{L_2}{T_2}$ and $\NFLEq{0}{L_1}{T_1}{L_2}{T_2}$
then $\FPB{h}{g}{L_1}{T_1}{L_2}{T_2}$.

\item\ObjLabel{lleq_fpb_trans}
(transitivity of proper $rst$-reduction through lazy equivalence)\\
If $\LLEq{0}{T}{K_1}{K_2}$ and $\FPB{h}{g}{K_2}{T}{L_2}{U}$ then\\
there exists $L_1$ such that
$\FPB{h}{g}{K_1}{T}{L_1}{U}$ and $\LLEq{0}{U}{L_1}{L_2}$.

\item\ObjLabel{fpbq_fpbg_trans}
(transitivity of proper $qrst$-computation through $qrst$-reduction, left case)\\
If $\FPBQ{h}{g}{L_1}{T_1}{L}{T}$ and $\FPBG{h}{g}{L}{T}{L_2}{T_2}$
then $\FPBG{h}{g}{L_1}{T_1}{L_2}{T_2}$.

\item\ObjLabel{fpbs_fpbg_trans}
(transitivity of proper $qrst$-computation through $qrst$-computation, left case)\\
If $\FPBS{h}{g}{L_1}{T_1}{L}{T}$ and $\FPBG{h}{g}{L}{T}{L_2}{T_2}$
then $\FPBG{h}{g}{L_1}{T_1}{L_2}{T_2}$.

\end{enumerate}
\end{theorem}

\begin{proof}
\Thref{}{fpbs_inv_alt} is proved by induction on its premise
rearranging the $qrst$-steps with 
\Thref{cpx_props}{fqu_cpx_trans}, \Thref{cpx_props}{lpx_fqu_trans},
\Thref{cpx_props}{lpx_cpx_trans}, \Thref{lleq_props}{lleq_fqu_trans},
\Thref{lleq_props}{lleq_cpx_trans}, and \Thref{lleq_props}{lleq_lpx_trans}.
\Thref{}{fpb_fpbq_alt} is proved by cases on its premise.
\Thref{}{fpbq_inv_fpb_alt} is proved by cases on its first premise.
\Thref{}{lleq_fpb_trans} is proved cases on its second premise
with the help of
\Thref{lleq_props}{lleq_fqu_trans}, \Thref{lleq_props}{lleq_cpx_trans},
\Thref{lleq_props}{lleq_lpx_trans}, and \Thref{lleq_props}{cpx_lleq_conf_sn}.
\Thref{}{fpbq_fpbg_trans} is a corollary of
\Thref{}{fpbq_inv_fpb_alt} and \Thref{}{lleq_fpb_trans}.
\Thref{}{fpbs_fpbg_trans} is proved by induction on its first premise
with the help of \Thref{}{fpbq_fpbg_trans}.
\end{proof}

Notice that the reverse of \Thref{fpb_props}{fpbs_inv_alt} is straightforward.
Also notice that \Thref{fpb_props}{fpbs_fpbg_trans} implies the transitivity
of proper $qrst$-computation.
The ``right case'' of the transitivity, that is: 
$\FPBG{h}{g}{L_1}{T_1}{L}{T} \MAnd \FPBS{h}{g}{L}{T}{L_2}{T_2} \MImp
\FPBG{h}{g}{L_1}{T_1}{L_2}{T_2}$, comes immediately from the
transitivity of $qrst$-computation
(defined as a transitive closure in \SubRef{fpb_defs}).
Another important corollary of \Thref{fpb_props}{fpbs_fpbg_trans} is that
the relation $\FSB{h}{g}{L}{T}$ is generated by the next rule:
\[
\Rule{\MAll{L_2,T_2} \FPBG{h}{g}{L_1}{T_1}{L_2}{T_2} \MImp \FSB{h}{g}{L_2}{T_2}}
     {\mathrm{fsb}}{\FSB{h}{g}{L_1}{T_1}}
\EqLabel{fsb_alt}\]
The induction principle for $\FSB{h}{g}{L}{T}$ derived from this rule,
gives a very strong induction hypothesis that takes advantage of the generality
of proper $qrst$-computation (\DefRef{fpbg}{}).
We need such a strength to prove the preservation \Thref{preserve_props}{}.

\begin{theorem}[strongly $qrst$-normalizing closures]\hfil
\ObjLabel{fsb_props}
\begin{enumerate}

\item\ObjLabel{csx_fsb_fpbs}
(strong normalization implies strong $qrst$-normalization, general form)\\
If $\CSX{h}{g}{L_1}{T_1}$ and $\FPBS{h}{g}{L_1}{T_1}{L_2}{T_2}$
then $\FSB{h}{g}{L_2}{T_2}$.

\item\ObjLabel{aaa_fsb}
(very big tree theorem)\\
If $\AAA{L}{T}{A}$ then $\FSB{h}{g}{L}{T}$
for each $\Next{h}{}$ and $\Deg{h}{g}{}{}$.

\end{enumerate}
\end{theorem}

\begin{proof}
\Thref{}{csx_fsb_fpbs} is proved
by induction on its first premise and then
by induction on the proper
subclosures of $\CLOSURE{L_2}{T_2}$
by invoking \Thref{lsx_props}{csx_lsx}
and the the reverse of \Thref{fpb_props}{fpbs_inv_alt}.
\Thref{}{aaa_fsb} is a corollary of
\Thref{}{csx_fsb_fpbs} and \Thref{gcp_props}{aaa_csx}.
\end{proof}

\subsection{Results on Stratified Validity}
\SubLabel{snv_props}

The relevant properties of stratified validity and of its refinement are listed next.

\begin{theorem}[stratified validity and its refinement]\hfil
\ObjLabel{snv_props}
\begin{enumerate}

\item\ObjLabel{snv_fwd_aaa}
(inclusion of validity)\\
If $\SNV{h}{g}{L}{T}$ then there exists $A$ such that $\AAA{L}{T}{A}$.

\item\ObjLabel{snv_fwd_fsb}
(validity implies strong qrst-normalization)\\
If $\SNV{h}{g}{L}{T}$ then $\FSB{h}{g}{L}{T}$.

\item\ObjLabel{lsubsv_fwd_lsubd}
(first inclusion of refinement)\\
If $\LSubSV{h}{g}{L_1}{L_2}$ then $\LSubD{h}{g}{L_1}{L_2}$.

\item\ObjLabel{lsubsv_fwd_lsuba}
(second inclusion of refinement)\\
If $\LSubSV{h}{g}{L_1}{L_2}$ then $\LSubA{L_1}{L_2}$.

\item\ObjLabel{lsubsv_lstas_trans}
(transitivity of degree-guarded iterated static type assignment through refinement)\\
If $n \le d$ and $\DA{h}{g}{L_2}{T}{d}$
then $\LSubSV{h}{g}{L_1}{L_2}$ and $\LSTAS{h}{L_2}{T}{U_2}{n}$\\
imply $\LSTAS{h}{L_1}{T}{U_1}{n}$ and $\CPCS{L_1}{U_1}{U_2}$ for some $U_1$.

\item\ObjLabel{lsubsv_scpds_trans}
(transitivity of stratified decomposed computation through refinement)\\
If $\LSubSV{h}{g}{L_1}{L_2}$ and $\SCPDS{h}{g}{L_2}{T_1}{T_2}{n}$\\
then $\SCPDS{h}{g}{L_1}{T_1}{T}{n}$ and $\CPRS{L_1}{T_2}{T}$ for some $T$.

\item\ObjLabel{lsubsv_snv_trans}
(transitivity of validity through refinement)\\
If $\LSubSV{h}{g}{L_1}{L_2}$ and $\SNV{h}{g}{L_2}{T}$
then $\SNV{h}{g}{L_1}{T}$.

\end{enumerate}
\end{theorem}

\begin{proof}
\Thref{}{snv_fwd_aaa} is proved by induction on its premise
by invoking \Thref{aaa_props}{aaa_mono}
and \Thref{aaa_props}{cpx_lpx_aaa_conf}.
Here we see that
preservation of validity requires preservation of atomic arity.
\Thref{}{snv_fwd_fsb} is a corollary of
\Thref{}{snv_fwd_aaa} and of \Thref{fsb_props}{aaa_fsb}.
\Thref{}{lsubsv_fwd_lsubd} is proved by induction on its premise.
\Thref{}{lsubsv_fwd_lsuba} is proved by induction on its premise
by invoking \Thref{}{snv_fwd_aaa}, \Thref{aaa_props}{lsuba_aaa_conf},
\Thref{aaa_props}{aaa_mono}, and \Thref{aaa_props}{cpx_lpx_aaa_conf}
when \FigRef{lsubsv}{\mathrm{beta}} is considered.
\Thref{}{lsubsv_lstas_trans} is proved by induction on last premise,
by cases on its second premise, and then by cases on its third premise.
\Thref{da_props}{da_lstas} is invoked among other propositions
when \FigRef{lstas}{\mathrm{zero}} and \FigRef{lstas}{\mathrm{succ}}
are considered in the case of \FigRef{lsubsv}{\mathrm{beta}}.
\Thref{}{lsubsv_scpds_trans} is a corollary of
\Thref{}{lsubsv_fwd_lsubd}, \Thref{}{lsubsv_lstas_trans},
\Thref{cpr_props}{lsubr_cpr_trans}, \Thref{cprs_props}{cpcs_inv_cprs},
and \Thref{da_props}{lsubd_da_trans}.
\Thref{}{lsubsv_snv_trans} is proved
by induction on its second premise and by cases on its first premise,
by invoking \Thref{}{lsubsv_scpds_trans} and
\Thref{cprs_props}{cprs_conf}
when \FigRef{snv}{\mathrm{appl}} and \FigRef{snv}{\mathrm{cast}}
are considered.
\end{proof}

We introduce some abbreviations in the style of \cite{SPA2}
to state the preservation theorem.
With respect to \cite{SPA2},
our $\PVR{}{}{}{}$ is connected to his $\mathrm{CL}$, and 
our $\PT{}{}{}{}$ is connected to his $\mathrm{P}^*\mathrm{T}$.

\begin{definition}[preservation properties]
\ObjLabel{IH}
\FigRef{IH}{} defines four properties of the closure $\CLOSURE{L_1}{T_1}$
with respect to $\Next{h}{}$ and $\Deg{h}{g}{}{}$.
They are:
preservation of degree by reduction ($\PD{}{}{}{}$),
preservation of validity by reduction ($\PVR{}{}{}{}$),
preservation of validity by static type ($\PVT{}{}{}{}$),
and preservation of static type by reduction ($\PT{}{}{}{}$).
\qed
\end{definition}

\begin{Figure}{Preservation properties}{IH}
\begin{small}\begin{tabular}{lll}
$\PD{h}{g}{L_1}{T_1}$&is&
$(\SNV{h}{g}{L_1}{T_1}) \MAnd \MAll{L_2,T_2,d}$\\
&&$(\DA{h}{g}{L_1}{T_1}{d}) \MAnd (\CPR{L_1}{T_1}{T_2}) \MAnd
(\LPR{L_1}{L_2}) \MImp (\DA{h}{g}{L_1}{T_1}{d})$\\\\

$\PVR{h}{g}{L_1}{T_1}$&is&
$(\SNV{h}{g}{L_1}{T_1}) \MAnd \MAll{L_2,T_2}$\\
&&$(\CPR{L_1}{T_1}{T_2}) \MAnd
(\LPR{L_1}{L_2}) \MImp (\SNV{h}{g}{L_2}{T_2})$\\\\

$\PVT{h}{g}{L_1}{T_1}$&is&
$(\SNV{h}{g}{L_1}{T_1}) \MAnd \MAll{U_1,d,n}$\\
&&$n \le d \MAnd (\DA{h}{g}{L_1}{T_1}{d}) \MAnd
\LSTAS{h}{L_1}{T_1}{U_1}{n} \MImp (\SNV{h}{g}{L_1}{U_1})$\\\\

$\PT{h}{g}{L_1}{T_1}$&is&
$(\SNV{h}{g}{L_1}{T_1}) \MAnd \MAll{L_2,T_2,U_1,d,n}$\\
&&$n \le d \MAnd (\DA{h}{g}{L_1}{T_1}{d}) \MAnd \LSTAS{h}{L_1}{T_1}{U_1}{n} \MAnd
(\CPR{L_1}{T_1}{T_2}) \MAnd$\\
&&$(\LPR{L_1}{L_2}) \MImp
\MEx{U_2} \LSTAS{h}{L_2}{T_2}{U_2}{n} \MAnd \CPCS{L_2}{U_1}{U_2}$ 
\end{tabular}\end{small}
\end{Figure}

\begin{theorem}[preservation properties]\hfil
\ObjLabel{preserve_props}
\begin{enumerate}

\item\ObjLabel{da_cpr_lpr_aux}
(conditional preservation of degree by reduction)\\
$(\MAll{L_1,T_1} \FPBG{h}{g}{L}{T}{L_1}{T_1} \MImp \PD{h}{g}{L_1}{T_1})$ and\\
$(\MAll{L_1,T_1} \FPBG{h}{g}{L}{T}{L_1}{T_1} \MImp \PVR{h}{g}{L_1}{T_1})$ and\\
$(\MAll{L_1,T_1} \FPBG{h}{g}{L}{T}{L_1}{T_1} \MImp \PVT{h}{g}{L_1}{T_1})$ imply
$\PD{h}{g}{L}{T}$.

\item\ObjLabel{snv_cpr_lpr_aux}
(conditional preservation of validity by reduction)\\
$(\MAll{L_1,T_1} \FPBG{h}{g}{L}{T}{L_1}{T_1} \MImp \PD{h}{g}{L_1}{T_1})$ and\\
$(\MAll{L_1,T_1} \FPBG{h}{g}{L}{T}{L_1}{T_1} \MImp \PVR{h}{g}{L_1}{T_1})$ and\\
$(\MAll{L_1,T_1} \FPBG{h}{g}{L}{T}{L_1}{T_1} \MImp \PVT{h}{g}{L_1}{T_1})$ and\\
$(\MAll{L_1,T_1} \FPBG{h}{g}{L}{T}{L_1}{T_1} \MImp \PT{h}{g}{L_1}{T_1})$ imply
$\PVR{h}{g}{L}{T}$.

\item\ObjLabel{snv_lstas_aux}
(conditional preservation of validity by static type)\\
$(\MAll{L_1,T_1} \FPBG{h}{g}{L}{T}{L_1}{T_1} \MImp \PD{h}{g}{L_1}{T_1})$ and\\
$(\MAll{L_1,T_1} \FPBG{h}{g}{L}{T}{L_1}{T_1} \MImp \PVR{h}{g}{L_1}{T_1})$ and\\
$(\MAll{L_1,T_1} \FPBG{h}{g}{L}{T}{L_1}{T_1} \MImp \PVT{h}{g}{L_1}{T_1})$ and\\
$(\MAll{L_1,T_1} \FPBG{h}{g}{L}{T}{L_1}{T_1} \MImp \PT{h}{g}{L_1}{T_1})$ imply
$\PVT{h}{g}{L}{T}$.

\item\ObjLabel{lstas_cpr_lpr_aux}
(conditional preservation of static type by reduction)\\
$(\MAll{L_1,T_1} \FPBG{h}{g}{L}{T}{L_1}{T_1} \MImp \PD{h}{g}{L_1}{T_1})$ and\\
$(\MAll{L_1,T_1} \FPBG{h}{g}{L}{T}{L_1}{T_1} \MImp \PVR{h}{g}{L_1}{T_1})$ and\\
$(\MAll{L_1,T_1} \FPBG{h}{g}{L}{T}{L_1}{T_1} \MImp \PVT{h}{g}{L_1}{T_1})$ and\\
$(\MAll{L_1,T_1} \FPBG{h}{g}{L}{T}{L_1}{T_1} \MImp \PT{h}{g}{L_1}{T_1})$ imply
$\PT{h}{g}{L}{T}$.

\item\ObjLabel{snv_preserve}
(preservation theorem, general form)\\
If $\SNV{h}{g}{L}{T}$ then 
$\PD{h}{g}{L}{T}$ and $\PVR{h}{g}{L}{T}$ and $\PVT{h}{g}{L}{T}$ and $\PT{h}{g}{L}{T}$. 

\item\ObjLabel{snv_cprs_lpr}
(preservation of validity by computation)\\
If $\CPRS{L}{T_1}{T_2}$
then $\SNV{h}{g}{L}{T_1}$ implies $\SNV{h}{g}{L}{T_2}$.

\item\ObjLabel{lstas_cpcs_lpr}
(preservation of conversion by static type)\\
If $\SNV{h}{g}{L}{T_1}$ and $\SNV{h}{g}{L}{T_2}$
and $n \le d_1$ and $n \le d_2$
and $\DA{h}{g}{L}{T_1}{d_1}$ and $\DA{h}{g}{L}{T_2}{d_2}$
and $\LSTAS{h}{L}{T_1}{U_1}{n}$ and $\LSTAS{h}{L}{T_2}{U_2}{n}$
then $\CPCS{L}{T_1}{T_2}$ implies $\CPCS{L}{U_1}{U_2}$.

\end{enumerate}
\end{theorem}

\begin{proof}
\Thref{}{da_cpr_lpr_aux}, \Thref{}{snv_cpr_lpr_aux},
\Thref{}{snv_lstas_aux}, and \Thref{}{lstas_cpr_lpr_aux}
are proved by cases on $T$, and then by cases on the other premises.
When \FigRef{cpr}{\beta} is considered,
\Thref{}{da_cpr_lpr_aux} invokes
\Thref{da_props}{lsubd_da_trans} and \FigRef{lsubd}{\mathrm{beta}},
\Thref{}{snv_cpr_lpr_aux} invokes 
\Thref{snv_props}{lsubsv_snv_trans} and \FigRef{lsubsv}{\mathrm{beta}},
while \Thref{}{lstas_cpr_lpr_aux} invokes
\Thref{snv_props}{lsubsv_lstas_trans} and \FigRef{lsubsv}{\mathrm{beta}}.
Moreover \Thref{}{snv_cpr_lpr_aux} needs
\Thref{cprs_props}{cprs_conf} in the cases of \FigRef{cpr}{\mathrm{flat}}
(already noted by \cite{SPA2}) and of \FigRef{cpr}{\theta}, while
\Thref{}{lstas_cpr_lpr_aux} needs \Thref{da_props}{da_lstas}
in the case of \FigRef{cpr}{\delta}.
\Thref{}{snv_preserve} is proved
by induction on the proper $rst$-reducts of $\CLOSURE{L}{T}$
by invoking 
\Thref{}{da_cpr_lpr_aux}, \Thref{}{snv_cpr_lpr_aux},
\Thref{}{snv_lstas_aux}, and \Thref{}{lstas_cpr_lpr_aux}.
The induction is assured by \Thref{snv_props}{snv_fwd_fsb}
and by \Ruleref{fsb_alt}.
\Thref{}{snv_cprs_lpr} is proved
by induction on its first premise
by invoking $\PVR{}{}{}{}$
from \Thref{}{snv_preserve}. 
\Thref{}{lstas_cpcs_lpr} is a corollary of
\Thref{}{snv_cprs_lpr},
\Thref{cprs_props}{cpcs_inv_cprs}, and \Thref{lstas_props}{lstas_mono},
given $\PVT{}{}{}{}$ and $\PD{}{}{}{}$
from \Thref{}{snv_preserve}. 
\end{proof}

\Thref{preserve_props}{snv_preserve} sums up
the most significant propositions discussed in this article.

\section{Conclusion and Future Work}
\SecLabel{conclusion}

We presented in \SecRef{definitions}
a revised version of the formal system $\LD$ to be termed
``$\LD$ version 2A1'', and we proved in \SecRef{propositions}
that this calculus enjoys three relevant desired properties:
confluence of reduction (\Thref{cprs_props}{}),
strong normalization along $qrst$-computations (\Thref{gcp_props}{}),
and preservation of validity by reduction (\Thref{preserve_props}{}).

Notably, the matter of this article was entirely developed by the author
with the unavoidable help of the proof management system Matita of \cite{ARST11},
which mechanically validated the resulting formalization of \cite{lambdadeltaV2a} in full.
The development took 42 months, producing 143 definitions and 1416 propositions.
More data is available at $\LD$ Web site \URL{\LDHome}.

We wish to stress that, to our knowledge, we are presenting
as \Thref{fsb_props}{aaa_fsb}
the first fully machine-checked proof of the so-called
``big tree'' theorem \cite{SPAc4} for a calculus that includes $\LY$.
It is also important to point out that
the proof of this theorem is harder in $\LD$ than in $\LY$
since the latter system does not have environments.

The long time we needed to take $\LD$ to this stage,
played in favor of presenting the development as is,
while the revision of the calculus is far from being complete.
In particular the present treatment lacks the type assignment judgment
$\NTA{h}{L}{T}{U}$ and its desired properties found in \cite{lambdadeltaJ1a}.
Anyway, it is a design feature of $\LD$, the fact that a term is typed
iff it is valid, so the preservation theorem presented here 
is the crux of the ``subject reduction'' property of this judgment.

Moreover, we are interested in relating the present notion of validity,
based on an extended (\IE $\LY$-like) applicability condition,
with the one implied by \cite{lambdadeltaJ1a},
which is based on a restricted (\IE PTS-like) applicability condition
(see \SubRef{snv_defs}).
It might happen that every valid closure in the extended sense has
an $\eta$-equivalent formulation that is valid in the restricted sense.
We support this conjecture by noting that a typical case
in which we need extended validity, is the next: 
\[
\SNV{h}{g}{\LDec{\LDec{\LDec{L}{z}{\Sort{k}}}{y}{(\Abst{}{\Sort{k}}{\Sort{k}})}}{x}{y}}{\Appl{z}{x}}
\EqLabel{extended}\]
where named variables improve the readability.
If we $\eta$-expand $y$ (\IE the expected type of $x$) to $\Abst{w}{\Sort{k}}{\Appl{w}{y}}$, 
restricted validity suffices.

It is important to stress that the above transformation
looks like an $\eta$-expansion because of the notation,
but it might have a different logical meaning.
We see such a case considering Landau's ``Grundlagen der Analysis''
(GdA) formalized in the system $\AUTQE$ \cite{SPA3},
where Automath's unified binder \verb+[x:W]+ stands
either for $\Abst{x}{W}{}$, or for $\Prod{x}{W}{}$.
The GdA validates just in the extended sense because a situation like \EqRef{extended}
occurs in the definition of the constant \verb+ande2"l-r"+,
but four formal $\eta$-expansions assure its validity in the restricted sense as well. 
Each one takes an expected type \verb+b+, that is the $y$ of \EqRef{extended},
and turns it into \verb+[x:a]<x>b+ (\verb+<x>+ is our applicator $\Appl{x}{}$).
We must note that the expected type of \verb+b+ is \verb+[x:a]'prop'+,
whereas the expected type of \verb+[x:a]<x>b+ is \verb+'prop'+.
So this expansion is not type-preserving,
especially if we accept the statement of \cite{Bro11} on the GdA that
every unified binder of degree one stands for a $\Pi$.
This means that the expansion is indeed a $\Pi$-introduction.
Interestingly, \cite{Bro11} states that formal $\eta$-expansions,
whose logical meaning should be investigated, solve all
incompatibilities preventing the GdA to validate in a PTS. 

\Thref{snv_props}{snv_fwd_aaa} shows that a
valid closure can by typed by a simple type.
Using $\LD$ as a logical framework is not a priority,
but if we wish to do so (say, for validating the GdA),
we need the additional expressive power given by universes
(say, $\star$ in the $\LCube$, or \verb+'type'+ and \verb+'prop'+ in the GdA). 
However, adding universes to $\LD$ while preserving its properties
is challenging because the naive extension of $\LY$
with ``type inclusion'' (the device with which universes are built
in the Automath tradition) is not conservative, since
either confluence or uniqueness of types is lost. 

Other additions to $\LD$ we shall consider, include:
``global'' variables referred by level
(while the current variables referred by depth would be ``local''),
term-like environments with projections as we advocated in \cite{lambdadeltaJ1a}, 
and metavariables.
Furthermore, we are interested in improving multiple relocation (\DefRef{lifts}{}),
which we introduced since 
the set of the functions $\Lift{l}{m}{}{}$ is not closed for composition,
by considering the functions $\Lifts{\vect{c}}{}{}$ as primitive,
and by representing a multiple relocation more conveniently than with a list of pairs.
As the reader can see,
our long-term aim is to make $\LD$ a fully fledged and elegant type system
suitable for many purposes. 

\begin{acks}

I am grateful to A. Asperti, C. Sacerdoti Coen, and S. Solmi
for their constant support and for their valuable advices
on several issues related to $\LD$.
I thank the anonymous referees
for their helpful suggestions with which I could improve this text.
I wish to dedicate this work to A.D. Bonanno and R. Prazzoli
for the joyful moments we shared in these years during the development of $\LD$.

\end{acks}

\bibliography{lambdadelta,spa,helm,kamareddine,mtt,coq,references}

\newpage\appendix

\section{Summary of Notation}
\AppLabel{notation}
\setcounter{section}{1}

The ongoing revision of $\LD$ includes
an update of the notational conventions of \cite{lambdadeltaJ1a}.
This Appendix summarizes
the revised notation we introduced in \SecRef{definitions}.

\medskip\begin{raggedleft}\begin{tabular}{lll}
\DefNotation{A,B}{atomic arity}{aaa}
\DefNotation{\C}{reducibility candidate}{gcr}
\DefNotation{K,L}{environment}{tl}
\DefNotation{\R}{generic property on closures}{gcr}
\DefNotation{T,U,V,W}{term}{tl}
\DefNotation{\vect{V}}{list of arguments}{applv}
\DefNotation{c}{relocation pair}{lift}
\DefNotation{\vect{c}}{list of relocation pairs}{lifts}
\DefNotation{d}{degree}{da}
\DefNotation{\Deg{}{g}{}{}}{sort degree parameter}{da}
\DefNotation{\Next{h}{}}{sort hierarchy parameter}{lstas}
\DefNotation{i,j}{variable reference depth}{tl}
\DefNotation{k}{sort index}{tl}
\DefNotation{l}{relocation level}{lift}
\DefNotation{m}{relocation depth}{lift}
\DefNotation{n}{number of iterations}{lstas}
\DefNotation{\Impl{B}{A}}{functional atomic arity}{aaa}
\DefNotation{\CFun{\C_1}{\C_2}}{function subset}{cfun}
\DefNotation{\Append{K}{L}}{concatenation}{append}
\DefNotation{\LDef{L}{}{V}}{definition}{tl}
\DefNotation{\LDec{L}{}{W}}{declaration}{tl}
\DefNotation{\LREq{l}{m}{L_1}{L_2}}{ranged equivalence}{lreq}
\DefNotation{\LPR{L_1}{L_2}}{parallel reduction for environments}{lpr}
\DefNotation{\LPX{h}{g}{L_1}{L_2}}{extended parallel reduction for env.'s}{lpx}
\DefNotation{\CNX{h}{g}{L}{T}}{normal term for extended reduction}{cnx}
\DefNotation{\LPRS{L_1}{L_2}}{parallel computation for environments}{cprs}
\DefNotation{\LPXS{h}{g}{L_1}{L_2}}{extended parallel computation for env.'s}{cpxs}
\DefNotation{\CSX{h}{g}{L}{T}}{strongly norm. term for ext. reduction}{cnx}
\DefNotation{\SNV{h}{g}{L}{T}}{stratified validity}{snv}
\DefNotation{\SHNV{h}{g}{L}{T}{d}}{stratified higher validity}{lsubsv}
\DefNotation{\AAA{L}{T}{A}}{atomic arity assignment}{aaa}
\DefNotation{\CQR{L}{T_1}{T_2}}{sequential reduction}{transition}
\DefNotation{\CQX{h}{g}{L}{T_1}{T_2}}{extended sequential reduction}{xtransition}
\DefNotation{\CPR{L}{T_1}{T_2}}{parallel reduction for terms}{cpr}
\DefNotation{\CPX{h}{g}{L}{T_1}{T_2}}{extended parallel reduction for terms}{cpx}
\DefNotation{\CPRS{L}{T_1}{T_2}}{parallel computation for terms}{cprs}
\DefNotation{\CPXS{h}{g}{L}{T_1}{T_2}}{extended parallel computation for terms}{cpxs}
\DefNotation{\CPCS{L}{T_1}{T_2}}{parallel conversion for terms}{cprs}
\DefNotation{\LSTAS{h}{L}{T}{U}{n}}{iterated static type assignment}{lstas}
\DefNotation{\SCPDS{h}{g}{L}{T_1}{T_2}{n}}{stratified decomposed computation}{scpds}
\DefNotation{\SCPES{h}{g}{L}{T_1}{T_2}{n_1}{n_2}}{stratified decomposed conversion}{scpds}
\DefNotation{\DA{h}{g}{L}{T}{d}}{degree assignment}{da}
\DefNotation{\LSubR{L_1}{L_2}}{refinement for preservation of reduction}{lsubr}
\DefNotation{\LSubC{\R}{}{L_1}{L_2}}{refinement for reducibility}{lsubc}
\end{tabular}\end{raggedleft}
\newpage
\begin{raggedleft}\begin{tabular}{lll}
\DefNotation{\LSubA{L_1}{L_2}}{refinement for preserv. of atomic arity}{lsuba}
\DefNotation{\LSubD{h}{g}{L_1}{L_2}}{refinement for preservation of degree}{lsubd}
\DefNotation{\LSubSV{h}{g}{L_1}{L_2}}{refinement for preserv. of strat. validity}{lsubsv}
\DefNotation{\LLEq{l}{T}{L_1}{L_2}}{lazy equivalence for environments}{lleq}
\DefNotation{\LLOr{l}{T}{L_1}{L_2}{L}}{pointwise union}{llor}
\DefNotation{\TSTS{T_1}{T_2}}{same top structure}{tsts}
\DefNotation{\MIter{n}{h}}{iterated composition}{lstas}
\DefNotation{\FreeS{i}{l}{L}{U}}{hereditarily free variable}{frees}
\SecNotation{\Empty}{empty list}{definitions}
\DefNotation{\Null}{empty environment}{tl}
\DefNotation{\Full}{base atomic arity}{aaa}
\DefNotation{\Sort{k}}{sort}{tl}
\DefNotation{\LRef{i}}{variable reference}{tl}
\DefNotation{\Length{L}}{length}{length}
\DefNotation{\Abbr{}{V}{L}}{tail definition}{append}
\DefNotation{\Abbr{}{V}{T}}{abbreviation}{tl}
\DefNotation{\Abst{}{W}{L}}{tail declaration}{append}
\DefNotation{\Abst{}{W}{T}}{abstraction}{tl}
\DefNotation{\Appl{V}{T}}{application}{tl}
\DefNotation{\ApplV{\vect{V}}{}{T}}{multiple application}{applv}
\DefNotation{\Cast{W}{T}}{type annotation}{tl}
\DefNotation{\Lifts{\vect c}{T_1}{T_2}}{multiple relocation}{lifts}
\DefNotation{\Lift{l}{m}{T_1}{T_2}}{relocation}{lift}
\DefNotation{\LiftV{l}{m}{\vect{T}_1}{\vect{T}_2}}{vector relocation}{liftv}
\DefNotation{\Drops{\vect c}{L_1}{L_2}}{multiple drop}{drops}
\DefNotation{\Drop{l}{m}{L_1}{L_2}}{drop}{drop}
\DefNotation{\LSX{h}{g}{l}{T}{L}}{strongly norm. env. for ext. reduction}{lsx}
\DefNotation{\LCoSX{h}{g}{l}{L}}{strongly co-norm. env. for ext. reduction}{lcosx}
\DefNotation{\FSB{h}{g}{L}{T}}{strongly norm. closure for $rst$-reduction}{fsb}
\SubNotation{\CLOSURE{L}{T}}{closure}{closures}
\DefNotation{\InGCRV{L}{\vect{V}}{\R}}{multiple habitation}{gcr}
\DefNotation{\FLEq{l}{L_1}{T_1}{L_2}{T_2}}{lazy equivalence for closures}{fsb}
\DefNotation{\FQU{L_1}{T_1}{L_2}{T_2}}{direct subclosure}{fqu}
\DefNotation{\FQUQ{L_1}{T_1}{L_2}{T_2}}{reflexive direct subclosure}{fqu}
\DefNotation{\FQUS{L_1}{T_1}{L_2}{T_2}}{subclosure}{fqus}
\DefNotation{\FPB{h}{g}{L_1}{T_1}{L_2}{T_2}}{proper $rst$-reduction}{fpbg}
\DefNotation{\FPBQ{h}{g}{L_1}{T_1}{L_2}{T_2}}{$qrst$-reduction}{fpbs}
\DefNotation{\FPBG{h}{g}{L_1}{T_1}{L_2}{T_2}}{proper $qrst$-computation}{fpbg}
\DefNotation{\FPBS{h}{g}{L_1}{T_1}{L_2}{T_2}}{$qrst$-computation}{fpbs}
\DefNotation{\CLOSURE{l}{m}}{relocation pair}{lift}
\DefNotation{\ACR{\R}{}{A}}{interpretation of the atomic arity}{acr}
\DefNotation{\Simple{T}}{simple (or neutral) term}{simple}
\DefNotation{\PD{h}{g}{L_1}{T_1}}{preservation of degree by reduction}{IH}
\DefNotation{\PT{h}{g}{L_1}{T_1}}{preservation of static type by reduction}{IH}
\DefNotation{\PVR{h}{g}{L_1}{T_1}}{preservation of validity by reduction}{IH}
\DefNotation{\PVT{h}{g}{L_1}{T_1}}{preservation of validity by static type}{IH}
\SecNotation{\Cons{}{}}{list concatenation}{definitions}
\SecNotation{\MAll{},\MEx{},\MImp,\MAnd}{metalinguistic logical constants}{definitions}
\DefNotation{\EITHER{}{}}{shared notation}{tl}
\SecNotation{\hbox{\hss\qed}}{end of definition, end of proof}{introduction}
\end{tabular}\end{raggedleft}

\newpage\section{Pointers to the Certified Proofs}
\AppLabel{pointers}

At the the moment of writing this article,
the certified specification of the revised $\LD$ is available
just as a bundle of script files for the latest version of the
proof management system Matita.
The bundle is available on the Web at:
\URL{\LDHome download/lambdadelta\_2A1.tar.gz}.
For each proposition stated in the article, we give a pointer
consisting of a path with three components:
a directory inside the directory \verb+basic_2+ of the bundle,
a file name inside this directory,
and a proved statement inside this file.

Notice that the notation in the files
may differ from \AppRef{notation} because of incompatibilities between
the characters available for \LaTeX{} and for Matita.

Moreover, the given pointers might be modified
in the forthcoming revisions of $\LD$.

Here we are referring to the directory
\URL{trunk/matita/matita/contribs/lambdadelta/}
in revision 12964 of HELM repository
at \URL{http://helm.cs.unibo.it/software/index.html}. 

\begin{enumerate}

\Pointer{cpr_props}{lsubr_trans}{static/lsubr\_lsubr/lsubr\_trans}
\Pointer{cpr_props}{lsubr_cpr_trans}{reduction/cpr/lsubr\_cpr\_trans}
\Pointer{cpr_props}{cpr_conf_lpr}{reduction/lpr\_lpr/cpr\_conf\_lpr}
\Pointer{cpr_props}{lpr_conf}{reduction/lpr\_lpr/lpr\_conf}

\Pointer{cprs_props}{cprs_conf}{computation/cprs\_cprs/cprs\_conf}
\Pointer{cprs_props}{lprs_conf}{computation/lprs\_lprs/lprs\_conf}
\Pointer{cprs_props}{cpcs_inv_cprs}{equivalence/cpcs\_cpcs/cpcs\_inv\_cprs}

\Pointer{da_props}{da_lstas}{unfold/lstas\_da/da\_lstas}
\Pointer{da_props}{lstas_inv_da}{unfold/lstas\_da/lstas\_inv\_da}
\Pointer{da_props}{lstas_inv_da_ge}{unfold/lstas\_da/lstas\_inv\_da\_ge}
\Pointer{da_props}{lsubd_fwd_lsubr}{static/lsubd/lsubd\_fwd\_lsubr}
\Pointer{da_props}{lsubd_da_trans}{static/lsubd\_da/lsubd\_da\_trans}
\Pointer{da_props}{lsubd_da_conf}{static/lsubd\_da/lsubd\_da\_conf}
\Pointer{da_props}{lsubd_trans}{static/lsubd\_lsubd/lsubd\_trans}

\Pointer{lstas_props}{lstas_mono}{unfold/lstas\_lstas/lstas\_mono}
\Pointer{lstas_props}{lstas_inv_refl_pos}{unfold/lstas\_da/lstas\_inv\_refl\_pos}

\Pointer{cpx_props}{lsubr_cpx_trans}{reduction/cpx/lsubr\_cpx\_trans}
\Pointer{cpx_props}{cpr_cpx}{reduction/cpx/cpr\_cpx}
\Pointer{cpx_props}{sta_cpx}{reduction/cpx\_lift/sta\_cpx}
\Pointer{cpx_props}{fqu_cpx_trans}{reduction/cpx\_lift/fqu\_cpx\_trans}
\Pointer{cpx_props}{lpx_fqu_trans}{reduction/lpx\_drop/lpx\_fqu\_trans}
\Pointer{cpx_props}{lpx_cpx_trans}{computation/lpxs\_lpxs/lpx\_cpx\_trans}
\Pointer{cpx_props}{cpxs_fwd_beta}{computation/cpxs\_tsts/cpxs\_fwd\_beta}

\Pointer{aaa_props}{lsuba_fwd_lsubr}{static/lsuba/lsuba\_fwd\_lsubr}
\Pointer{aaa_props}{lsuba_aaa_trans}{static/lsuba\_aaa/lsuba\_aaa\_trans}
\Pointer{aaa_props}{lsuba_aaa_conf}{static/lsuba\_aaa/lsuba\_aaa\_conf}
\Pointer{aaa_props}{lsuba_trans}{static/lsuba\_lsuba/lsuba\_trans}
\Pointer{aaa_props}{aaa_mono}{static/aaa\_aaa/aaa\_mono}
\Pointer{aaa_props}{aaa_lstas}{unfold/lstas\_aaa/aaa\_lstas}
\Pointer{aaa_props}{cpx_lpx_aaa_conf}{reduction/lpx\_aaa/cpx\_lpx\_aaa\_conf}

\Pointer{gcp_props}{csx_gcr}{computation/csx\_tsts\_vector/csx\_gcr}
\Pointer{gcp_props}{acr_gcr}{computation/gcp\_cr/acr\_gcr}
\Pointer{gcp_props}{acr_aaa_csubc_lifts}{computation/gcp\_aaa/acr\_aaa\_csubc\_lifts}
\Pointer{gcp_props}{gcr_aaa}{computation/gcp\_aaa/gcr\_aaa}
\Pointer{gcp_props}{aaa_csx}{computation/csx\_aaa/aaa\_csx}
\Pointer{gcp_props}{lsubc_fwd_lsubr}{computation/lsubc/lsubc\_fwd\_lsubr}
\Pointer{gcp_props}{lsuba_lsubc}{computation/lsubc\_lsuba/lsuba\_lsubc}

\Pointer{llor_props}{llor_tail_frees}{multiple/llor\_alt/llor\_tail\_frees}
\Pointer{llor_props}{llor_tail_cofrees}{multiple/llor\_alt/llor\_tail\_cofrees}
\Pointer{llor_props}{llor_total}{multiple/llor\_drop/llor\_total}

\Pointer{lleq_props}{llpx_sn_llor_fwd_sn}{multiple/llpx\_sn\_llor/llpx\_sn\_llor\_fwd\_sn}
\Pointer{lleq_props}{llpx_sn_llor_dx}{multiple/lleq\_llor/llpx\_sn\_llor\_dx}
\Pointer{lleq_props}{lleq_lreq_trans}{multiple/lleq\_lreq/lleq\_lreq\_trans}
\Pointer{lleq_props}{lleq_fqu_trans}{multiple/lleq\_fqus/lleq\_fqu\_trans}
\Pointer{lleq_props}{lleq_cpx_trans}{reduction/cpx\_lleq/lleq\_cpx\_trans}
\Pointer{lleq_props}{lleq_lpx_trans}{reduction/lpx\_lleq/lleq\_lpx\_trans}
\Pointer{lleq_props}{cpx_lleq_conf_sn}{reduction/cpx\_lleq/cpx\_lleq\_conf\_sn}

\Pointer{lsx_props}{lsx_cpx_trans_lcosx}{computation/lcosx\_cpx/lsx\_cpx\_trans\_lcosx}
\Pointer{lsx_props}{lsx_lref_be_lpxs}{computation/lsx\_csx/lsx\_lref\_be\_lpxs}
\Pointer{lsx_props}{csx_lsx}{computation/lsx\_csx/csx\_lsx}

\Pointer{fpb_props}{fpbs_inv_alt}{computation/fpbs\_alt/fpbs\_inv\_alt}
\Pointer{fpb_props}{fpb_fpbq_alt}{reduction/fpbq\_alt/fpb\_fpbq\_alt}
\Pointer{fpb_props}{fpbq_inv_fpb_alt}{reduction/fpbq\_alt/fpbq\_inv\_fpb\_alt}
\Pointer{fpb_props}{lleq_fpb_trans}{reduction/fpb\_lleq/lleq\_fpb\_trans}
\Pointer{fpb_props}{fpbq_fpbg_trans}{computation/fpbg\_fpbs/fpbq\_fpbg\_trans}
\Pointer{fpb_props}{fpbs_fpbg_trans}{computation/fpbg\_fpbs/fpbs\_fpbg\_trans}

\Pointer{fsb_props}{csx_fsb_fpbs}{computation/fsb\_csx/csx\_fsb\_fpbs}
\Pointer{fsb_props}{aaa_fsb}{computation/fsb\_aaa/aaa\_fsb}

\Pointer{snv_props}{snv_fwd_aaa}{dynamic/snv\_aaa/snv\_fwd\_aaa}
\Pointer{snv_props}{snv_fwd_fsb}{dynamic/snv\_fsb/snv\_fwd\_fsb}
\Pointer{snv_props}{lsubsv_fwd_lsubd}{dynamic/lsubsv\_lsubd/lsubsv\_fwd\_lsubd}
\Pointer{snv_props}{lsubsv_fwd_lsuba}{dynamic/lsubsv\_lsuba/lsubsv\_fwd\_lsuba}
\Pointer{snv_props}{lsubsv_lstas_trans}{dynamic/lsubsv\_lstas/lsubsv\_lstas\_trans}
\Pointer{snv_props}{lsubsv_scpds_trans}{dynamic/lsubsv\_scpds/lsubsv\_scpds\_trans}
\Pointer{snv_props}{lsubsv_snv_trans}{dynamic/lsubsv/lsubsv\_snv\_trans}

\Pointer{preserve_props}{da_cpr_lpr_aux}{dynamic/snv\_da\_lpr/da\_cpr\_lpr\_aux}
\Pointer{preserve_props}{snv_cpr_lpr_aux}{dynamic/snv\_lpr/snv\_cpr\_lpr\_aux}
\Pointer{preserve_props}{snv_lstas_aux}{dynamic/snv\_lstas/snv\_lstas\_aux}
\Pointer{preserve_props}{lstas_cpr_lpr_aux}{dynamic/snv\_lstas\_lpr/lstas\_cpr\_lpr\_aux}
\Pointer{preserve_props}{snv_preserve}{dynamic/snv\_preserve/snv\_preserve}
\Pointer{preserve_props}{snv_cprs_lpr}{dynamic/snv\_preserve/snv\_cprs\_lpr}
\Pointer{preserve_props}{lstas_cpcs_lpr}{dynamic/snv\_preserve/lstas\_cpcs\_lpr}

\end{enumerate}

\end{document}